\documentclass[a4paper,UKenglish,cleveref, autoref, thm-restate]{lipics-v2021}
\nolinenumbers

\usepackage[linesnumbered,titlenotnumbered,ruled]{algorithm2e}
\usepackage{tikz,pgffor}
\usetikzlibrary{arrows,backgrounds}
\usepackage{tkz-euclide}
\usetikzlibrary{arrows}
\usetikzlibrary{shapes}
\usetikzlibrary{calc}
\usetikzlibrary{automata}
\usetikzlibrary{positioning}
\usepackage{listings}
\usepackage{transparent}
\usepackage{calrsfs}
\DeclareMathAlphabet{\pazocal}{OMS}{zplm}{m}{n}

\newcommand{\N}{\mathbb{N}}

\newcommand{\Z}{\mathbb{Z}}
\newcommand{\A}{\pazocal{A}}
\newcommand{\C}{\pazocal{C}}
\newcommand{\D}{\pazocal{D}}

\newcommand{\B}{\pazocal{B}}

\newcommand{\U}{\pazocal{U}}
\newcommand{\V}{\pazocal{V}}
\newcommand{\calG}{\pazocal{G}}
\newcommand{\calC}{\pazocal{C}}
\newcommand{\calL}{\pazocal{L}}
\newcommand{\cL}{\mathcal{L}}
\newcommand{\calO}{\pazocal{O}}
\newcommand{\calE}{\pazocal{E}}

\newcommand{\calP}{\pazocal{P}}

\newcommand{\ce}[1]{\left(#1 \right)}

\newcommand{\eff}[1]{\Delta(#1)}

\newcommand{\Tval}{\mathit{Tval}}
\newcommand{\Cval}{\mathit{Cval}}

\newcommand{\Tran}{\textit{Tran}}
\newcommand{\Comp}{\mathit{Comp}}

\newcommand{\Aux}{\mathit{Aux}}
\newcommand{\Degree}{\textit{Deg}}
\newcommand{\len}{\textit{len}}
\newcommand{\mmax}{\textit{max}}
\newcommand{\Vect}{\textit{Vectors}}
\newcommand{\Pre}{\textit{Pre}}

\newcommand{\Gad}{\textit{Gadget}}

\newcommand{\bu}{\mathbf{u}}
\newcommand{\bv}{\mathbf{v}}
\newcommand{\bw}{\mathbf{w}}

\newcommand{\PTIME}{\mathbf{P}}
\newcommand{\AP}{\mathbf{AP}}
\newcommand{\NP}{\mathbf{NP}}
\newcommand{\coNP}{\mathbf{coNP}}
\newcommand{\DP}{\mathbf{DP}}
\newcommand{\PSPACE}{\mathbf{PSPACE}}
\renewcommand{\leftarrow}{\texttt{\,+=\,}}
\newcommand{\tran}[1]{\xrightarrow{#1}}

\tikzstyle{label}=[shape=circle,draw,inner sep=0pt,minimum size=5mm]
\tikzstyle{line} = [draw, color=black, -latex]
\tikzstyle{max}=[thick,draw,color=mybrown,minimum size=1.4em,inner sep=0em]
\tikzstyle{min}=[diamond,thick,draw,color=mybrown,minimum size=1.4em,%
inner sep=0em]
\tikzstyle{stoch}=[circle,thick,draw,minimum size=1.5em,%
inner sep=0em]
\tikzstyle{ran}=[circle,thick,draw,minimum size=1.4em,%
inner sep=0em]
\tikzstyle{mc}=[rounded corners,thick,draw,minimum size=1.4em,%
inner sep=.5ex]
\tikzstyle{tran}=[thick,draw,->,>=stealth,rounded corners]
\tikzstyle{loop left}=[tran, to path={.. controls +(150:.8) 
	and +(210:.8) .. (\tikztotarget) \tikztonodes}]
\tikzstyle{loop right}=[tran, to path={.. controls +(30:.8) 
	and +(330:.8) .. (\tikztotarget) \tikztonodes}]
\tikzstyle{loop above}=[tran, to path={.. controls +(60:.5) 
	and +(120:.5) .. (\tikztotarget) \tikztonodes}]
\tikzstyle{loop below}=[tran, to path={.. controls +(240:.8) 
	and +(300:.8) .. (\tikztotarget) \tikztonodes}]
\tikzstyle{bigstoch}=[circle,draw,minimum size=8ex,inner sep=0pt,font=\Large, very thick,text centered, fill=blue!20, fill opacity=0.2, draw=black!80, text opacity=1]
\tikzstyle{bigstoch2}=[circle,draw,minimum size=8ex,inner sep=0pt,font=\Large, very thick,text centered, fill=blue!20, fill opacity=0.2, draw=black!80, text opacity=1, double]
\tikzstyle{bigmin}=[diamond,draw,minimum size=8ex,inner sep=0pt,font=\Large, very thick,text centered, fill=blue!20, fill opacity=0.2, draw=black!80, text opacity=1]
\tikzstyle{bigtran}=[very thick,draw,-angle 60,font=\scriptsize, inner sep = 6pt]

\def\Inf{\operatornamewithlimits{inf\vphantom{p}}}

\title{Deciding Polynomial Termination Complexity for VASS Programs} 

\titlerunning{Polynomial VASS Termination} 

\author{Michal Ajdar{\'{o}}w}{Masaryk University, Czechia \and \url{https://www.muni.cz/lide/422654-michal-ajdarow} }{}{}{}

\author{Anton\'{\i}n Ku\v{c}era}{Masaryk University, Czechia \and \url{https://www.fi.muni.cz/usr/kucera/}}{tony@fi.muni.cz}{https://orcid.org/0000-0002-6602-8028}{}

\authorrunning{M.~Ajdar{\'{o}}w and A.~Ku\v{c}era} 

\Copyright{Michal Ajdar{\'{o}}w and Anton\'{\i}n Ku\v{c}era} 



\category{} 

\funding{The work is supported by the Czech Science Foundation, Grant No.~21-24711S.}

\EventEditors{Serge Haddad and Daniele Varacca}
\EventNoEds{2}
\EventLongTitle{32nd International Conference on Concurrency Theory (CONCUR 2021)}
\EventShortTitle{CONCUR 2021}
\EventAcronym{CONCUR}
\EventYear{2021}
\EventDate{August 23--27, 2021}
\EventLocation{Virtual Conference}
\EventLogo{}
\SeriesVolume{203}
\ArticleNo{24}

\begin{document}

\maketitle

\begin{abstract}
We show that for every fixed degree $k\geq 3$, the problem whether the termination/counter complexity of a given demonic VASS is $\calO(n^k)$, $\Omega(n^{k})$, and $\Theta(n^{k})$ is $\coNP$-complete, $\NP$-complete, and \mbox{$\DP$-complete}, respectively. We also classify the complexity of these problems for $k\leq 2$. This shows that the polynomial-time algorithm designed for strongly connected demonic VASS in previous works cannot be extended to the general case. Then, we prove that the same problems for VASS games are $\PSPACE$-complete. Again, we classify the complexity also for $k\leq 2$. 
Tractable subclasses of demonic VASS and VASS games are obtained by bounding certain structural parameters, which opens the way to applications in program analysis despite the presented lower complexity bounds.

\end{abstract}

\section{Introduction}
\label{sec-intro}

Vector addition systems with states (VASS) are a generic formalism expressively equivalent to Petri nets. In program analysis, VASS are used to model programs with unbounded integer variables, parameterized systems, etc. Thus, various problems about such systems reduce to the corresponding questions about VASS. This approach's main bottleneck is that interesting questions about VASS tend to have high computational complexity (see, e.g., \cite{CLLLM:VASS-reach-nonelem,Lipton:PN-Reachability,MM:containment-Petri}). Surprisingly, recent results (see below) have revealed computational tractability of problems related to \emph{asymptotic complexity} of VASS computations, allowing to answer questions like ``Does the program terminate in time polynomial in~$n$ for all inputs of size $n$?'', or ``Is the maximal value of a given variable bounded by $\calO(n^4)$ for all inputs of size~$n$?''. These results are encouraging and may enhance the existing software tools for asymptotic program analysis such as SPEED~\cite{SPEED2}, COSTA~\cite{DBLP:journals/entcs/AlbertAGGPRRZ09},
RAML~\cite{journals/toplas/0002AH12},
Rank~\cite{ADFG10:lexicographic-ranking-flowcharts},
Loopus~\cite{SZV14,journals/jar/SinnZV17},
AProVE~\cite{journals/jar/GieslABEFFHOPSS17},
CoFloCo~\cite{conf/aplas/Flores-MontoyaH14},
C4B~\cite{conf/cav/Carbonneaux0RS17}, and others. In this paper, we give a full classification of the computational complexity of deciding polynomial termination/counter complexity for demonic VASS and VASS games, and solve open problems formulated in previous works. Furthermore, we identify structural parameters making the asymptotic VASS analysis computationally hard. Since these parameters are often small in VASS program abstractions, this opens the way to applications in program analysis despite the established lower complexity bounds.

The \emph{termination complexity} of a given VASS $\A$ is a function $\calL : \N \rightarrow \N_\infty$ assigning to every $n$ the maximal length of a computation initiated in a configuration with all counters initialized to~$n$. Similarly, the \emph{counter complexity} of a given counter $c$ in $\A$ is a function $\calC[c] : \N \rightarrow \N_\infty$ such that $\calC[c](n)$ is the maximal value of $c$ along a computation initiated in a configuration with all counters set to~$n$. So far, three types of VASS models have been investigated in previous works.
\begin{itemize}
	\item \emph{Demonic VASS}, where the non-determinism is resolved by an adversarial environment aiming to increase the complexity.
	\item \emph{VASS Games}, where every control state is declared as \emph{angelic} or \emph{demonic}, and the non-determinism is resolved by the controller or by the environment aiming to lower and increase the complexity, respectively.
	\item \emph{VASS MDPs}, where the states are either non-deterministic or stochastic. The non-determinism is usually resolved in the ``demonic'' way.
\end{itemize}
Let us note that the ``angelic'' and ``demonic'' non-determinism are standard concepts in program analysis  \cite{BW:nondet-languages} applicable to arbitrary computational devices including VASS. The use of VASS termination/counter complexity analysis is illustrated in the next example.


\begin{example}
	\label{exa-demonic}
	Consider the program skeleton of Fig.~\ref{fig-VASS-model}~(left). Since a VASS cannot directly model the assignment \texttt{j:=i} and cannot test a counter for zero, the skeleton is first transformed into an equivalent program of Fig.~\ref{fig-VASS-model}~(middle), where the assignment \texttt{j:=i} is implemented using an auxiliary variable \texttt{Aux} and two \texttt{while} loops. Clearly, the execution of the transformed program is only longer than the execution of the original skeleton (for all inputs). For the transformed program, an over-approximating demonic VASS model is obtained by replacing conditionals with non-determinism, see Fig.~\ref{fig-VASS-model}~(right). When all counters are initialized to~$n$, the VASS terminates after $\calO(n^2)$ transitions. Hence, the same upper bound is valid also for the original program skeleton. Actually, the run-time complexity of the skeleton is $\Theta(n^2)$ where $n$ is the initial value of~\texttt{i}, so the obtained upper bound is asymptotically optimal.
\end{example}

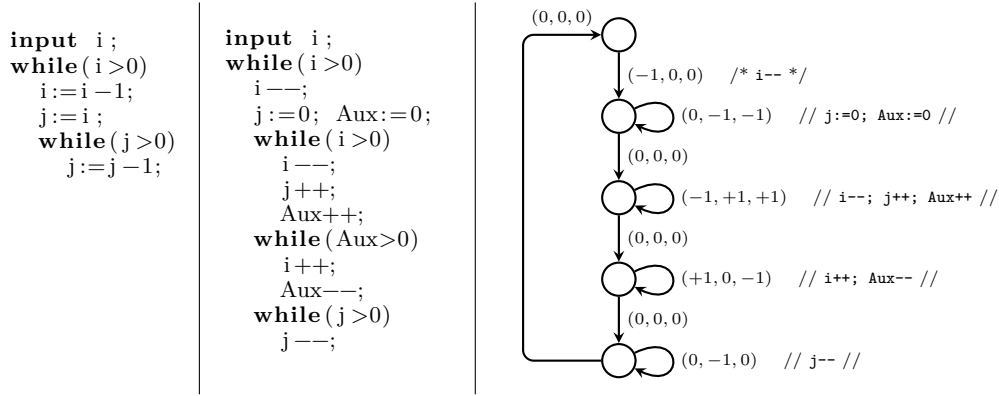
\begin{figure}[t]
	\lstset{basicstyle=\footnotesize}
	\begin{tabular}{l|l|l}
		\begin{minipage}{2.5cm}
			\begin{lstlisting}
input i;
while(i>0) 
  i:=i-1;
  j:=i;
  while(j>0)
    j:=j-1;
			\end{lstlisting}
			\vspace*{7em}
		\end{minipage}
		&
		\begin{minipage}{3.3cm}
			\begin{lstlisting}
input i;
while(i>0) 
  i--;
  j:=0; Aux:=0;
  while(i>0)
    i--; 
    j++;
    Aux++;
  while(Aux>0)
    i++;
    Aux--;
  while(j>0)
    j--;
			\end{lstlisting}
		\end{minipage}
		&
		\begin{minipage}{4cm}\hspace*{1em}
			\begin{tikzpicture}[scale=.9, every node/.style={scale=0.9}, x=1.4cm, y=1.2cm, font=\scriptsize]
				\foreach \x in {0,1,2,3,4}{ 
					\node [stoch] (t\x) at (0,-\x)  {};
				}
				\foreach \x/\y in {1/2,2/3,3/4}{
					\draw [tran] (t\x) to node[right] {$(0,0,0)$} (t\y);
				}
				\draw [tran] (t0) to node[right] {$(-1,0,0)$ \quad /* \texttt{i{-}-} */ } (t1);   
				\draw [tran,rounded corners] (t4) -- +(-1,0) --   +(-1,4) -- node[above] {$(0,0,0)$} (t0);   
				\draw [tran,loop right] (t1) to node[right] {$(0,-1,-1)$  \quad // \texttt{j:=0; Aux:=0} // } (t1); 
				\draw [tran,loop right] (t2) to node[right] {$(-1,+1,+1)$ \quad // \texttt{i{-}-; j++; Aux++} //} (t2); 
				\draw [tran,loop right] (t3) to node[right] {$(+1,0,-1)$ \quad // \texttt{i++; Aux{-}-} //} (t3); 
				\draw [tran,loop right] (t4) to node[right] {$(0,-1,0)$ \quad // \texttt{j{-}-} //} (t4); 
			\end{tikzpicture}
		\end{minipage}
	\end{tabular}
	\caption{A skeleton of a simple imperative program (left) and its VASS model (right).}
	\label{fig-VASS-model}
\end{figure}


\noindent
\textbf{Existing results.}
In \cite{BCKNVZ:VASS-linear-termination}, it is shown that the problem whether $\calL \in \calO(n)$ for a given demonic VASS is solvable in polynomial time, and a complete proof method based on linear ranking functions is designed. The polynomiality of termination complexity for a given demonic VASS is also decidable in polynomial time, and if $\calL \not\in \calO(n^k)$ for any $k \in \N$, then $\calL \in 2^{\Omega(n)}$ \cite{Leroux:Polynomial-termination-VASS}. The same results hold for counter complexity. In \cite{Zuleger:VASS-polynomial}, a polynomial time algorithm computing the \emph{least} $k \in \N$ such that $\calL \in \calO(n^k)$ for a given demonic VASS is presented (the algorithm first checks if such a $k$ exists). It is also shown that if $\calL \not\in \calO(n^k)$, then $\calL \in \Omega(n^{k+1})$. Again, the same results hold also for counter complexity. The proof is actually given only for \emph{strongly connected demonic VASS}, and it is conjectured that a generalization to unrestricted demonic VASS can be obtained by extending the presented construction (see the Introduction of \cite{Zuleger:VASS-polynomial}). In \cite{KLV:VASS-Grzegorczyk}, it was shown that the problem whether the termination/counter complexity of a given demonic VASS belongs to a given level of Grzegorczyk hierarchy is solvable in polynomial time, and the same problem for VASS games is shown $\NP$-complete. The $\NP$ upper bound follows by observing that player Angel can safely commit to a  \emph{countreless}\footnote{A strategy is \emph{counterless} if the decision depends just on the control state of the configuration currently visited.} strategy when minimizing the complexity level in the Grzegorczyk hierarchy. Intuitively, this is because Grzegorczyk classes are closed under function composition (unlike the classes $\Theta(n^k)$).  Furthermore, the problem whether $\calL \in \calO(n^2)$ for a given VASS game is shown $\PSPACE$ hard, but the decidability of this problem is left open. As for VASS MDPs, the only existing result is \cite{BCKNV:probVASS-linear-termination}, where it is shown that the linearity of termination complexity is solvable in polynomial time for VASS MDPs with a tree-like MEC decomposition.
\smallskip

\noindent
\textbf{Our contribution.}  For demonic VASS, we \emph{refute} the conjecture of \cite{Zuleger:VASS-polynomial} and prove that for general (not necessarily strongly connected) demonic VASS, the problem whether
\begin{itemize}
	\item $\calL \in \calO(n^k)$ is in $\PTIME$ for $k = 1$, and $\coNP$-complete for $k \geq 2$;
	\item $\calL \in \Omega(n^{k})$ is in $\PTIME$ for $k \leq 2$, and $\NP$-complete for $k \geq 3$;
	\item $\calL \in \Theta(n^{k})$ is in $\PTIME$ for $k =1$, $\coNP$-complete for $k=2$, and $\DP$-complete for $k \geq 3$;
	\item $\calC[c] \in \calO(n^k)$ is $\coNP$-complete;
	\item $\calC[c] \in \Omega(n^{k})$ is in $\PTIME$ for $k = 1$, and $\NP$-complete for $k \geq 2$;
	\item $\calC[c] \in \Theta(n^{k})$ is $\coNP$-complete for $k=1$, and $\DP$-complete for $k \geq 2$.
\end{itemize}

Since the demonic VASS constructed in our proofs are relatively complicated, we write them in a simple imperative language with a precisely defined VASS semantics. This allows to present the overall proof idea clearly and justify technical correctness by following the control flow of the VASS program, examining possible side effects of the underlying ``gadgets'', and verifying that the Demon does not gain anything by deviating from the ideal execution scenario.

When proving the upper bounds, 
we show that every path in the DAG of strongly connected components can be associated with the (unique) vector describing the maximal simultaneous increase of the counters. Here, the counters pumpable to exponential (or even larger) values require special treatment. 
We show that this vector is computable in polynomial time. Hence, the complexity of a given counter $c$ is  $\Omega(n^k)$ iff there is a path in the DAG such that the associated maximal increase of $c$ is $\Omega(n^k)$. Thus, we obtain the $\NP$ upper bound, and the other upper bounds follow similarly. The crucial parameter characterizing hard-to-analyze instances is the number of different paths from a root to a leaf in the DAG decomposition, and tractable subclasses of demonic VASS are obtained by bounding this parameter. We refer to Section~\ref{sec-demonic} for more details.

Then, we turn our attention to VASS games, where the problem of polynomial termination/counter complexity analysis requires completely new ideas. In \cite{KLV:VASS-Grzegorczyk}, it was observed that the information about the ``asymptotic counter increase performed so far'' must be taken into account by player Angel when minimizing the complexity level in the polynomial hierarchy, and counterless strategies are therefore insufficient. However, it is not clear what information is needed to make optimal decisions, and whether this information is finitely representable. We show that player Angel can safely commit to a so-called \emph{locking} strategy. A strategy for player Angel is \emph{locking} if whenever a new angelic state $p$ is visited, one of its outgoing transition is chosen and ``locked'' so that 
when $p$ is revisited, the same locked transition is used. The locked transition choice may depend on the computational history and the transitions locked in previously visited angelic states. Then, we define a \emph{locking decomposition} of a given VASS that plays a role similar to the DAG decomposition for demonic VASS. Using the locking decomposition, the existence of a suitable locking strategy for player Angel is decided by an alternating polynomial time algorithm (and hence in polynomial space). Thus, we obtain the following: For every VASS game and a counter \(c \), we have that $\calL, \calC[c]$ are either in $\calO(n^k)$ or in $\Omega(n^{k+1})$. Furthermore, the problem whether
\begin{itemize}
	\item $\calL \in \calO(n^k)$ is $\NP$-complete for $k {=} 1$ and $\PSPACE$-complete for $k{\geq} 2$;
	\item $\calL \in \Omega(n^{k})$ is in $\PTIME$ for $k {=} 1$, $\coNP$-complete for $k{=}2$, and $\PSPACE$-complete for $k{\geq} 3$;
	\item $\calL \in \Theta(n^{k})$ is $\NP$-complete for $k{=}1$ and $\PSPACE$-complete for $k{\geq} 2$;
	\item $\calC[c] \in \calO(n^k)$ is $\PSPACE$-complete;
	\item $\calC[c] \in \Omega(n^{k})$ is in $\PTIME$ for $k {=} 1$, and $\PSPACE$-complete for $k{\geq} 2$;
	\item $\calC[c] \in \Theta(n^{k})$ is $\PSPACE$-complete.
\end{itemize}
Similarly to demonic VASS, tractable subclasses of VASS games are obtained by bounding the number of different paths in the locking decomposition. 

The VASS model constructed in Example~\ref{exa-demonic} is purely demonic. The use of VASS \emph{games} in program analysis/synthesis is illustrated in the next example.

\begin{example}
	Consider the program of Fig.~\ref{fig-VASS-game-model}. The \texttt{condition} at line~3 is resolved by the environment in a demonic way. The two branches of $\texttt{if-then-else}$ execute a code modifying the variables \texttt{j} and \texttt{k}. After that, the controller can choose one of the two \textbf{while}-loops at lines~8, 9 with the aim of keeping the value of \texttt{z} small. The question is how the size of \texttt{z} grows with the size of input if the controller makes optimal decisions.
	A closer look reveals that when the variable \texttt{i} is assigned $n$ at line~1, then 
	\begin{itemize}
		\item the values of \texttt{j} and \texttt{k} are $\Theta(n)$ and $\Theta(n^2)$ when the \texttt{condition} is evaluated to \textit{true};
		\item the values of \texttt{j} and \texttt{k} are $\Theta(n^2)$ and $\Theta(n)$ when the \texttt{condition} is evaluated to \textit{false}.
	\end{itemize}
	Hence, the controller can keep $\texttt{z}$ in $\Theta(n)$ if an optimal decision is taken.
	Constructing a VASS game model for the program of Fig.~\ref{fig-VASS-game-model} is straightforward (the required gadgets are given in Fig.~\ref{fig-gadgets}). Using the results of this paper, the above analysis can be performed \emph{fully automatically}.
\end{example}

\begin{figure}
	\lstset{numbers=left,basicstyle=\footnotesize}
	\begin{lstlisting}
input i;
j:=0; k:=0; z:=0;
if condition   // demonic choice //
then while (i>0) do j++; k:=k+i; i--; done
else j:=i*i; k:=i;
while (i>0) do j:=j+k; i--; done
choose:        // angelic choice //
while (j>0) do j--; z++ done
or: while (k>0) do k--; z++ done
	\end{lstlisting}
	\caption{A simple program with both demonic and angelic non-determinism.}
	\label{fig-VASS-game-model}
\end{figure}

\section{Preliminaries}
\label{sec-prelim}

The sets of integers and non-negative integers are denoted by $\Z$ and $\N$, respectively, and we use $\N_\infty$ to denote $\N \cup \{\infty\}$.
The vectors of $\Z^d$ where $d \geq 1$ are denoted by $\bv,\bu,\ldots$, and the vector $(n,\ldots,n)$ is denoted by $\vec{n}$. 




\begin{definition}[VASS]
	\label{def-VASS}
	Let $d \geq 1$. A \emph{$d$-dimensional vector addition system with states
		(VASS)} is a pair $\A = \ce{Q,\Tran}$, where $Q \neq \emptyset$ is a finite set of
	\emph{states} and  $\Tran \subseteq Q \times \Z^d \times Q$ is a finite set of
	\emph{transitions} such that for every $q \in Q$ there exist $p \in Q$
	and $\bu \in \Z^d$ such that $(q,\bu,p) \in \Tran$.
\end{definition}

The set $Q$ is split into two disjoint subsets $Q_A$ and $Q_D$ of \emph{angelic} and \emph{demonic} states controlled by the players Angel and Demon, respectively. A \emph{configuration} of $\A$ is a pair $p \bv \in Q \times \N^d$, where $\bv$ is the vector of counter values. We often refer to counters by their symbolic names. For example, when we say that $\A$ has three counters $x,y,z$ and the value of $x$ in a configuration $p\bv$ is $8$, we mean that $d=3$ and $\bv_i = 8$ where $i$ is the index associated to $x$. When the mapping between a counter name and its index is essential, we use $c_i$ to denote the counter with index~$i$.  

A \emph{finite path} in $\A$ of length $m$ is a finite sequence $\varrho = p_1,\bu_1,p_2,\bu_2, \ldots,p_m$ such that $(p_i,\bu_{i},p_{i+1}) \in Tran$ for all \mbox{$1\leq i<m$}. We use $\eff{\varrho}$ to denote the \emph{effect} of $\varrho$, defined as $\sum_{i = 1}^m \bu_i$. An \emph{infinite path} in $\A$ is an infinite sequence $\alpha = p_1,\bu_1,p_2,\bu_2, \ldots$ such that every finite prefix $p_1,\bu_1,\ldots,p_m$ of $\alpha$ is a finite path in~$\A$.

%
A \emph{computation} of $\A$ is a sequence of configurations $\alpha = p_1\bv_1,p_2\bv_2,\ldots$ of length $m \in \N_\infty$ such that for every $1 \leq i < m$ there is a transition $(p_i,\bu_i,p_{i+1})$ satisfying $\bv_{i+1} = \bv_i + \bu_i$. Note that every computation determines its associated path in the natural way.

\smallskip

\noindent
\textbf{VASS Termination Complexity.}
A \emph{strategy} for Angel (or Demon) in $\A$ is a function $\eta$ assigning to every finite computation $p_1 \bv_1,\ldots, p_m \bv_{m}$ where $p_m \in Q_A$ (or $p_m \in Q_D$) a transition $(p_m,\bu, q)$. 
Every pair of strategies $(\sigma,\pi)$ for Angel/Demon and every initial configuration $p\bv$ determine the unique \emph{maximal} computation $\Comp^{\sigma,\pi}(p\bv)$ initiated in $p\bv$. The maximality means that the computation cannot be prolonged without making some counter negative. For a given counter $c$, we use $\mmax[c](\Comp^{\sigma,\pi}(p\bv))$ to denote the supremum of the $c$'s values in all configurations visited along $\Comp^{\sigma,\pi}(p\bv)$. Furthermore, we use $\len(\Comp^{\sigma,\pi}(p\bv))$ to denote the length of $\Comp^{\sigma,\pi}(p\bv)$. Note that $\mmax[c]$ and $\len$ can be infinite for certain computations.  


For every initial configuration $p\bv$, consider a game where the players Angel and Demon aim at minimizing and maximizing the $\mmax[c]$ or $\len$ objective, respectively. By applying standard game-theoretic arguments (see Appendix \ref{app-determinacy}), we obtain
\begin{eqnarray}
	\sup_{\pi} \ \Inf_{\sigma} \ \len(\Comp^{\sigma,\pi}(p\bv))  & = &  \Inf_{\sigma} \ \sup_{\pi}\   \len(\Comp^{\sigma,\pi}(p\bv))\label{eq-len-determinacy}\\[1ex]
	\sup_{\pi} \ \Inf_{\sigma} \ \mmax[c](\Comp^{\sigma,\pi}(p\bv))  & = &  \Inf_{\sigma} \ \sup_{\pi}\   \mmax[c](\Comp^{\sigma,\pi}(p\bv))
	\label{eq-c-determinacy}
\end{eqnarray}
where $\sigma$ and $\pi$ range over all strategies for Angel and~Demon, respectively.
Hence, there exists a unique \emph{termination value} of $p\bv$, denoted by $\Tval(p\bv)$, defined by~\eqref{eq-len-determinacy}. Similarly, for every counter $c$ there exists a unique \emph{maximal counter value}, denoted by $\Cval[c](p\bv)$, defined by~\eqref{eq-c-determinacy}. Furthermore, both players have \emph{optimal} positional strategies $\sigma^*$ and $\pi^*$ achieving the outcome specified by the equilibrium value or better in every configuration $p\bv$ against every strategy of the opponent (here, a \emph{positional} strategy is a strategy depending only on the currently visited configuration). We refer to Appendix \ref{app-determinacy} for details. 

The \emph{termination complexity} and \emph{$c$-counter complexity} of $\A$ are functions $\calL,\calC[c]: \N \rightarrow \N_\infty$ where $\calL(n) =  \max \{\Tval(p\vec{n}) \mid p \in Q \}$ and 
$\calC[c](n)  =  \max \{\Cval[c](p\vec{n}) \mid p \in Q \}$.
When the underlying VASS $\A$ is not clear, we write $\calL_\A$ and $\calC_\A[c]$ instead of $\calL$ and $\calC[c]$.

Observe that the asymptotic analysis of termination complexity for a given VASS $\A$ is trivially reducible to the asymptotic analysis of counter complexity in a VASS $\B$ obtained from $\A$ by adding a fresh ``step counter'' \textit{sc} incremented by every transition of $\B$. Clearly, $\calL_\A \in \Theta (\calC_\B[\textit{sc}])$. Therefore, the lower complexity bounds for the considered problems of asymptotic analysis are proven for $\calL$, while the upper bounds are proven for $\calC[c]$.

\section{Demonic VASS}
\label{sec-demonic}

In this section, we classify the computational complexity of polynomial asymptotic analysis for demonic VASS. The following theorem holds regardless whether the counter update vectors are encoded in unary or binary (the lower bounds hold for unary encoding, the upper bounds hold for binary encoding).

\begin{theorem}
	\label{thm-main-demon}
	Let $k \geq 1$. For every demonic VASS $\A$ and counter \(c\)  we have that $\calL$ ($\calC[c]$) is either in $\calO(n^k)$ or in $\Omega(n^{k+1})$. Furthermore, the problem whether
	\begin{itemize}
		\item $\calL \in \calO(n^k)$ is in $\PTIME$ for $k = 1$, and $\coNP$-complete for $k \geq 2$;
		\item $\calL \in \Omega(n^{k})$ is in $\PTIME$ for $k \leq 2$, and $\NP$-complete for $k \geq 3$;
		\item $\calL \in \Theta(n^{k})$ is in $\PTIME$ for $k =1$, $\coNP$-complete for $k=2$, and $\DP$-complete for $k \geq 3$;
		\item $\calC[c] \in \calO(n^k)$ is $\coNP$-complete;
		\item $\calC[c] \in \Omega(n^{k})$ is in $\PTIME$ for $k = 1$, and $\NP$-complete for $k \geq 2$;
		\item $\calC[c] \in \Theta(n^{k})$ is $\coNP$-complete for $k=1$, and $\DP$-complete for $k \geq 2$.
	\end{itemize}
\end{theorem}

The next theorem identifies the crucial parameter influencing the complexity of polynomial asymptotic analysis for demonic VASS. Let $\D(\A)$ be the standard DAG of strongly connected components of~$\A$. For every leaf (bottom SCC) $\eta$ of $\D(\A)$, let $\Degree(\eta)$ be the total number of all paths from a root of $\D(\A)$ to $\eta$.


\begin{theorem}
	\label{thm-demon-tractable}
	Let $\Lambda$ be a class of demonic VASS such that for every $\A \in \Lambda$ and every leaf $\eta$ of $\D(\A)$ we have that $\Degree(\eta)$ is bounded by a fixed constant depending only on $\Lambda$. 
	
	Then, the problems whether $\calL_\A \in \calO(n^k)$, $\calL_\A \in \Omega(n^{k})$, $\calL_\A \in \Theta(n^{k})$ for given $\A \in \Lambda$ and $k \in \N$, are solvable in polynomial time (where the $k$ is written in binary). The same results hold also for $\calC[c]$ (for a given counter $c$ of $\A$).
\end{theorem}

The degree of the polynomial bounding the running time of the decision algorithm for the three problems of Theorem~\ref{thm-demon-tractable} increases with the increasing size of the constant bounding $\Degree(\eta)$. 
From the point of view of program analysis, Theorem~\ref{thm-demon-tractable} has a clear intuitive meaning.
If $\A$ is an abstraction of a program $\calP$, then the instructions in $\calP$ increasing the complexity of the asymptotic analysis of $\A$ are \emph{branching instructions} such as \textbf{if-then-else} that are \emph{not embedded within loops}. If $\calP$ executes many such constructs in a sequence, a termination point can be reached in many ways (``zigzags'' in the $\calP$'s control-flow graph).
This increases $\Degree(\eta)$, where $\eta$ is a leaf of $\D(\A)$ containing the control state modeling the termination point of~$\calP$. 


\subsection{Lower bounds}


\SetAlgorithmName{VASS Program}{vassprog}{VASS Programs}
\SetAlFnt{\footnotesize}
\SetInd{0ex}{1.55ex}
\begin{algorithm}[tb]
	\SetAlgoLined
	\DontPrintSemicolon
	\SetKwInOut{Parameter}{parameter}\SetKwInOut{Input}{input}\SetKwInOut{Output}{output}
	\SetKwData{C}{C}\SetKwData{D}{D}\SetKwData{MX}{M}\SetKwData{f}{f}
	\SetKw{Choose}{choose:}	
	\SetKw{Or}{or}
	\SetKw{And}{and}
	\BlankLine
	$d_2 \leftarrow d_1 * e_1;\ \ d_3 \leftarrow d_2 * e_2; \ \cdots \ ; \ \ d_k \leftarrow d_{k-1} * e_{k-1};$ \;\label{dk-init}
	\ForEach{$i = 1,\ldots,v$\label{asign-start}}{
		$\Choose$\ \ $x_i  \leftarrow d_k$ \Or  $\bar{x}_i  \leftarrow d_k;$\;
	}\label{asign-end}
	$s_0  \leftarrow d_k;$\;\label{line-s0}
	\ForEach{$i = 1,\ldots,m$}{
		$\Choose$\ \ $s_i \leftarrow \min(\ell^i_1,s_{i-1})$
		\Or $s_i \leftarrow \min(\ell^i_2,s_{i-1})$
		\Or $s_i \leftarrow \min(\ell^i_3,s_{i-1});$\;
	}\label{line-end-loop}
	$f \leftarrow s_m * n$\label{line-fin}
	\caption{$\A_\varphi$}
	\label{alg-Aphi}
\end{algorithm}

\begin{figure}[tb]\centering
	\begin{tikzpicture}[scale=.85, every node/.style={scale=0.85}, x=2.1cm, y=2.1cm, font=\footnotesize]
		\node[stoch] (n1) at (0,0) {in};
		\node[stoch] (n2) at (1,0) {};
		\node[stoch] (n3) at (.5,-.7) {out};
		\draw [tran,loop above] (n1) to node[above] {\parbox{3em}{\centering -- $x$\\+ $\alpha$\\+ $z$}} (n1); 
		\draw [tran,loop above] (n2) to node[above] {\parbox{3em}{\centering + $x$\\-- $\alpha$\\+ $z$}} (n2); 
		\draw [tran,<->] (n1) to node[above] {-- $y$} (n2);  
		\draw [tran] (n1) -- (n3);
		\draw [tran] (n2) -- (n3);
		\node at (.5,-1.1) {$z \leftarrow x * y$};  
		\begin{scope}[shift={(2.5,0)}]
			\node[stoch] (n1) at (0,0) {in};
			\node[stoch] (n2) at (1,0) {out};
			\draw [tran,loop above] (n1) to node[above] {\parbox{3em}{\centering -- $y$\\+ $x$\\+ $\alpha$}} (n1); 
			\draw [tran,loop above] (n2) to node[above] {\parbox{3em}{\centering + $y$\\-- $\alpha$}} (n2);
			\draw [tran] (n1) to (n2);
			\node at (.5,-.3) {$x \leftarrow y$};
		\end{scope}
		\begin{scope}[shift={(5,0)}]
			\node[stoch] (n1) at (0,0) {in};
			\node[stoch] (n2) at (1,0) {out};
			\draw [tran,loop above] (n1) to node[above] 
			{\parbox{3em}{\centering -- $s_{i-1}$\\-- $\ell$\\+ $s_{i}$\\+ $\alpha$}} (n1); 
			\draw [tran,loop above] (n2) to node[above] 
			{\parbox{3em}{\centering + $\ell$\\-- $\alpha$}} (n2); 
			\draw [tran] (n1) to (n2);
			\node at (.5,-.3) {$s_i \leftarrow \min(\ell,s_{i-1})$};
		\end{scope}
		\begin{scope}[shift={(0,-1.9)}]
			\node[stoch] (n1) at (0,0)    {in};
			\node[stoch] (n2) at (0.5,0)  {in};
			\node[stoch] (n3) at (1.25,0) {out};
			\node[stoch] (n4) at (1.75,0) {in};
			\node[stoch] (n5) at (2.5,0)  {out};
			\node[stoch] (n6) at (3,0)    {out};
			\draw [tran] (n1) to (n2);
			\draw [tran] (n3) to (n4);
			\draw [tran] (n5) to (n6);
			\draw [rounded corners,draw] (.25,-.3) rectangle (1.45,.2); 
			\draw [rounded corners,draw] (1.55,-.3) rectangle (2.7,.2); 
			\node at (0.9,-.2) {$\textit{ins}_1$};
			\node at (2.15,-.2) {$\textit{ins}_2$};
			\node at (1.5,-.6) {$\textit{ins}_1;\ \textit{ins}_2$};
		\end{scope}
		\begin{scope}[shift={(4.5,-1.6)}]
			\node[stoch] (n1) at (0,0)    {in};
			\node[stoch] (n2) at (0.5,.5)  {in};
			\node[stoch] (n3) at (1.25,.5) {out};
			\node[stoch] (n4) at (0.5,-.5) {in};
			\node[stoch] (n5) at (1.25,-.5){out};
			\node[stoch] (n6) at (1.75,0)  {out};
			\draw [tran] (n1) to (n2);
			\draw [tran] (n1) to (n4);
			\draw [tran] (n3) to (n6);
			\draw [tran] (n5) to (n6);
			\draw [dotted,thick] (.9,-.1) -- (.9,.1);
			\draw [rounded corners,draw] (.25,.2)  rectangle (1.5,.8); 
			\draw [rounded corners,draw] (.25,-.8) rectangle (1.5,-.2); 
			\node at (0.9,.3)  {$\textit{ins}_1$};
			\node at (0.9,-.7) {$\textit{ins}_j$};
			\node at (.9,-1)  {$\textbf{choose: } \textit{ins}_1; \textbf{ or } \cdots  \textbf{ or } \textit{ins}_j$};
		\end{scope}
	\end{tikzpicture}
	\caption{The gadgets of $\A_\varphi$.}
	\label{fig-gadgets}
\end{figure}

\begin{lemma}
	\label{lem-hardness1}
	Let $k \geq 2$. For every propositional formula $\varphi$ in 3-CNF there exists a demonic VASS $\A_\varphi$, with counter \(c \), constructible in time polynomial in $|\varphi|$ such that 
	\begin{itemize}
		\item if $\varphi$ is satisfiable, then $\calL_{\A_\varphi} \in \Theta(n^{k+1})$ and $\calC_{\A_\varphi}[c] \in \Theta(n^k)$;
		\item if $\varphi$ is not satisfiable, then $\calL_{\A_\varphi} \in \Theta(n^{k})$ and $\calC_{\A_\varphi}[c] \in \Theta(n)$.
	\end{itemize}
\end{lemma}
\begin{proof}
	Let $\varphi \ \equiv \ C_1 \wedge \cdots \wedge C_m$ be a propositional formula where every $C_i \equiv \ell_1^i \vee \ell_2^i \vee \ell_3^i$ is a clause with three literals over propositional variables $X_1,\ldots,X_v$ (a literal is a propositional variable or its negation). We construct a VASS $\A_\varphi$ with the counters 
	\begin{itemize}
		\item $x_1,\cdots,x_v$, $\bar{x}_1,\cdots,\bar{x}_v$ used to encode an assignment of truth values to $X_1,\ldots,X_v$. In the following, we identify literals $\ell_j^i$ of $\varphi$ with their corresponding counters (i.e., if $\ell_j^i \equiv X_u$, the corresponding counter is $x_u$; and if $\ell_j^i \equiv \neg X_u$, the corresponding counter is $\bar{x}_u$).
		\item $s_0,\ldots,s_m=c$ used to encode the validity of clauses under the chosen assignment,
		\item $f$ used to encode the (in)validity of $\varphi$ under the chosen assignment,
		\item $d_1,\ldots,d_k$ and $e_1,\ldots,e_{k-1}$ used to compute $n^k$,
		\item and some auxiliary counters used in gadgets.
	\end{itemize}
	The structure of $\A_\varphi$ is shown in VASS Program~\ref{alg-Aphi}. The basic instructions are implemented by the gadgets of Fig.~\ref{fig-gadgets}~(top). Counter changes associated to a given transition are indicated by the corresponding labels, where $-c$ and $+c$ mean decrementing and incrementing a given counter by one (the other counters are unchanged). Hence, the empty label represents no counter change, i.e., the associated counter update vector is $\vec{0}$. The auxiliary counter $\alpha$ is \emph{unique for every instance} of these gadgets and it is not modified anywhere else.   
	
	The constructs  $\textit{ins}_1;\, \textit{ins}_2$ and 
	$\textbf{choose: } \textit{ins}_1; \textbf{ or } \cdots  \textbf{ or } \textit{ins}_j$ are implemented by connecting the underlying gadgets as shown in Fig.~\ref{fig-gadgets}~(bottom). The \textbf{foreach} statements are just concise representations of the corresponding sequences of instructions connected by `;'.
	
	Now suppose that the computation of VASS Program~\ref{alg-Aphi} is executed from line~1 where all counters are initialized to~$n$. One can easily verify that all gadgets implement the operations associated to their labels up to some ``asymptotically irrelevant side effects''. More precisely,  
	\begin{itemize}
		\item the $z \leftarrow x * y$ gadget ensures that the Demon can increase the value of counter $z$ by $\textit{val}(x) + \textit{val}(y) \cdot (\textit{val}(x) +n)$ (but not more) if he plays optimally, where $\textit{val}(x)$ and $\textit{val}(y)$ are the values stored in $x$ and $y$ when initiating the gadget. Recall that the counter $\alpha$ is unique for the gadget, and its initial value is~$n$. Also note that the value of $y$ is decreased to $0$ when the Demon strives to maximally increase the value of~$z$.
		\item The $x \leftarrow y$ gadget ensures that the Demon can add $\textit{val}(y)$ to the counter $x$ and then reset $y$ to the value $\textit{val}(y) + n$ (but not more) if he plays optimally.  Again, note that $\alpha$ is a unique counter for the gadget with initial value~$n$.
		\item The $s_i \leftarrow \min(\ell,s_{i-1})$ gadget allows the Demon to increase $s_i$ by the minimum of $\textit{val}(\ell)$ and $\textit{val}(s_{i-1})$, and then restore $\ell$ to  $\textit{val}(\ell) + n$ (but not more). 
	\end{itemize}
	Now, the VASS Program~\ref{alg-Aphi} is easy to follow. We describe its execution under the assumption that the Demon plays \emph{optimally}. It is easy to verify that the Demon cannot gain anything by deviating from the below described scenario where certain counters are pumped to their \emph{maximal} values (in particular, the auxiliary counters are never re-used outside their gadgets, hence the Demon is not motivated to leave any positive values in them). 
	
	By executing line~\ref{dk-init}, the Demon pumps the counter $d_k$ to the value $\Theta(n^k)$. Then, the Demon determines a truth assignment for every $X_i$, where $i \in \{1,\ldots,v\}$, by pumping either the counter $x_i$ or the counter $\bar{x}_i$ to the value $\Theta(n^k)$.
	A key observation is that when the chosen assignment makes $\varphi$ true, then every clause contains a literal such that the value of its associated counter is $\Theta(n^k)$. Otherwise, there is a clause $C_i$ such that all of the three counters corresponding to $\ell_1^i$, $\ell_2^i$, $\ell_3^i$ have the value~$n$.
	The Demon continues by pumping $s_0$ to the value $\Theta(n^k)$ at line~\ref{line-s0}. Then, for every $i = 1,\ldots,m$, he selects a literal $\ell_j^i$ of $C_i$ and pumps $s_i$ to the minimum of $\textit{val}(s_{i-1})$ and $\textit{val}(\ell_j^i)$. Observe that $\textit{val}(s_{i-1})$ is either $\Theta(n)$ or $\Theta(n^k)$, and the same holds for $\textit{val}(s_i)$ after executing the instruction. Hence, $s_m$ is pumped either to $\Theta(n^k)$ or  $\Theta(n)$, depending on whether the chosen assignment sets every clause to true or not, respectively. Note that the length of the whole computation up to line~\ref{line-fin} is $\Theta(n^k)$, regardless whether the chosen assignment sets the formula $\varphi$ to true or false. If $s_m$ was pumped to $\Theta(n^k)$, then the last instruction at line~\ref{line-fin} can pump the counter $f$ to $\Theta(n^{k+1})$ in $\Theta(n^{k+1})$ transitions. Hence, if $\varphi$ is satisfiable, the Demon can schedule a computation of length $\Theta(n^{k+1})$ and along the way pump \(c=s_m \) to \(\Theta(n^k)\). Otherwise, the length of the longest computation is $\Theta(n^{k})$ and counter \(c=s_m \) is bounded by $\Theta(n)$. Also observe that if the Demon starts executing $\A_{\varphi}$ in some other control state (i.e., not in the first instruction of line~1), the maximal length of a computation as well as the maximal value of \(c\) is only smaller.   
\end{proof}

Recall that the class $\DP$ consists of problems that are intersections of one problem in $\NP$ and another problem in $\coNP$. The class $\DP$ is expected to be somewhat larger than $\NP \cup \coNP$, and it is contained in the $\PTIME^{\NP}$ level of the polynomial hierarchy. The standard $\DP$-complete problem is \textsc{Sat-Unsat}, where an instance is a pair $\varphi,\psi$ of propositional formulae and the question is whether $\varphi$ is satisfiable and $\psi$ is unsatisfiable \cite{Papadimitriou:book}. Hence, the $\DP$ lower bounds of Theorem~\ref{thm-main-demon} follow directly from the next lemmas (a proof can be found in Appendix \ref{app-demonic}).


\begin{restatable}{lemma}{LemDPhard}
	\label{lem-hardness2}	
	Let $k \geq 3$. For every pair $\varphi,\psi$ of propositional formulae in 3-CNF there exists a demonic VASS $\A_{\varphi,\psi}$ such that $\calL_{\A_{\varphi,\psi}} \in \Theta(n^k)$ iff $\varphi$ is satisfiable and $\psi$ is unsatisfiable.
\end{restatable}

\begin{restatable}{lemma}{LemDPhardcounter}
	\label{lem-hardness3}	
	Let $k \geq 2$. For every pair $\varphi,\psi$ of propositional formulae in 3-CNF there exists a demonic VASS $\A_{\varphi,\psi}$ along with a counter \( c\) such that $\calC_{\A_{\varphi,\psi}}[c] \in \Theta(n^{k})$ iff $\varphi$ is satisfiable and $\psi$ is unsatisfiable.
\end{restatable}

\subsection{Upper bounds}
\label{sec-upper}

The upper complexity bounds of Theorem~\ref{thm-main-demon} are proven for $\calC[c]$. For the sake of clarity, we first sketch the main idea and then continue with developing a formal proof.
\smallskip

\noindent
\textbf{Intuition.}
For a given demonic VASS $\A$, we compute its SCC decomposition and proceed by analyzing the individual SCCs in the way indicated in Fig.~\ref{fig-decomp}. We start in a top SCC with all counters initialized to~$n$. Here, we can directly apply the results of~\cite{Zuleger:VASS-polynomial,Leroux:Polynomial-termination-VASS} and decide in polynomial time  whether $\calC[c] \in \Theta(n^k)$ for some $k \in \N$ or $\calC[c] \in 2^{\Omega(n)}$ (in the first case, we can also determine the $k$). We perform this analysis for every counter $c$ and thus obtain the vector describing the maximal asymptotic growth of the counters (such as $(n^2,n,2^{\Omega(n)})$ in Fig.~\ref{fig-decomp}). Observe that
\begin{itemize}
	\item although the asymptotic growth of $\calC[c]$ has been analyzed for each counter independently, all counters can be pumped to their associated asymptotic values \emph{simultaneously}. Intuitively, this is achieved by considering the ``pumping computations'' for the smaller vector $(\lfloor n/d \rfloor ,\ldots,\lfloor n/d \rfloor)$ of initial counter values ($d$ is the dimension of $\A$), and then simply ``concatenating'' these computations in a configuration with all counters initialized to~$n$; 
	\item if the asymptotic growth of $\calC[c]$ is $\Theta(n)$, the computation simultaneously pumping the counters to their asymptotic values may actually \emph{decrease} the value of $c$ (the computation can be arranged so that the resulting value of $c$ stays above  $\lfloor n/d \rfloor$ for all sufficiently large~$n$). For example, the top SCC of Fig.~\ref{fig-decomp} achieves the simultaneous asymptotic growth of all counters from $(n,n,n)$ to $(n^2,n,2^{\Omega(n)})$, but this does \emph{not} imply the counters can be simultaneously increased above the original value~$n$ (nevertheless, the simultaneous increase in the first and the third counter above $n$ is certainly possible for all sufficiently large~$n$).  
\end{itemize}
A natural idea how to proceed with next SCCs  is to perform a similar analysis for larger vectors of initial counter values. Since we are interested just in the asymptotic growth of the counters, we can safely set the initial value of a counter previously pumped to $\Theta(n^k)$ to \emph{precisely} $n^k$. However, it is not immediately clear how to treat the counters previously pumped to $2^{\Omega(n)}$. We show that  the length of a computation ``pumping'' the counters to their new asymptotic values in the considered SCC $\calC$ is at most exponential in $n$. Consequently, the ``pumping computation'' in $\calC$ can be constructed so that the resulting value of the ``large'' counters stays above one half of their original value. This means the value of ``large'' counters is still in $2^{\Omega(n)}$ after completing the computation in $\calC$. Furthermore, the large counters can be treated as if their initial value was infinite when analyzing $\calC$. This ``infinite'' initial value is implemented simply by modifying every counter update vector $\bu$ in $\calC$ so that $\bu[j] = 0$ for every ``large'' counter $c_j$. This adjustment in the structure of $\calC$ is denoted by putting ``$\infty$'' into the corresponding component of the initial counter value vector (see Fig.~\ref{fig-decomp}). This procedure is continued until processing all SCCs. Note that the same SCC may be processed \emph{multiple times} for different vectors of initial counter values corresponding to different paths from a top SCC. In Fig.~\ref{fig-decomp}, the bottom SCC is processed for the initial vectors $(n^5,n,\infty)$ and $(n^2,\infty,\infty)$ corresponding to the two paths from the top SCC. The number of such initial vectors can be \emph{exponential} in the size of $\A$, as witnessed by the VASS constructed in the proof of  Lemma~\ref{lem-hardness1}.

%
%
\begin{figure}
	\centering
	\begin{tikzpicture}[x=1.5cm, y=1.5cm, font=\footnotesize]
		\node (A) at (0,0) [draw,thick,minimum width=4cm,minimum height=1.5cm, rounded corners] {};
		\node (B) at (1.5,1.5) [draw,thick,minimum width=3cm,minimum height=1.5cm, rounded corners] {};
		\node (C) at (-1.5,1.5) [draw,thick,minimum width=3cm,minimum height=1.5cm, rounded corners] {};
		\node (D) at (0,3) [draw,thick,minimum width=3cm,minimum height=1.5cm, rounded corners] {};
		\draw [->, thick] (B) -- (A);
		\draw [->, thick] (C) -- (A);
		\draw [->, thick] (D) -- (B);
		\draw [->, thick] (D) -- (C);
		\node (a) at (0,3.3) {$(n,n,n)$};
		\node (b) at (0,2.7) {$(n^2,n,2^{\Omega(n)})$};
		\node (c) at (1.5,1.85) {$(n^2,n,\infty)$};
		\node (d) at (1.5,1.15)    {$(n^2,2^{\Omega(n)},\infty)$};
		\node (c) at (-1.5,1.85) {$(n^2,n,\infty)$};
		\node (d) at (-1.5,1.15)    {$(n^5,n,\infty)$};
		\node (c) at (-.6,.35) {$(n^5,n,\infty)$};
		\node (d) at (-.6,-.3)    {$(n^5,n^5,\infty)$};
		\node (c) at (.6,.35) {$(n^2,\infty,\infty)$};
		\node (d) at (.6,-.3)    {$(n^2,\infty,\infty)$};
	\end{tikzpicture}
	\caption{Analyzing $\calC[c]$ in a demonic VASS by SCC decomposition.}
	\label{fig-decomp}
\end{figure}
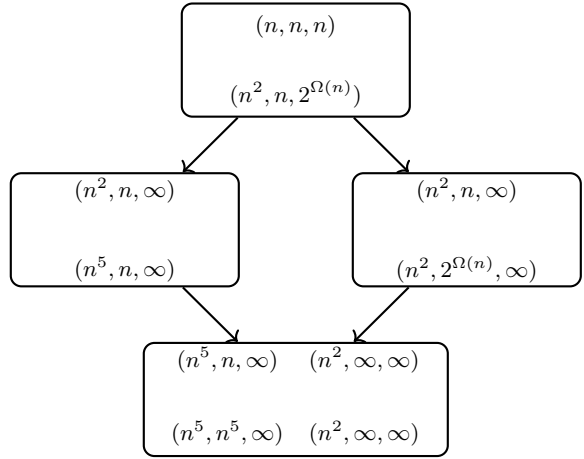

\smallskip

Now we give a formal proof. Let $\A$ be a demonic VASS with $d$ counters. For every counter $c$ and every $\bv \in \N^d$, we define the function $\C[c,\bv] : 
\N \rightarrow \N_\infty$ where $\C[c,\bv](n)$ is the maximum of all $\Cval[c](p\bu)$ where $p\in Q$ and $\bu = (n^{\bv(1)},\ldots,n^{\bv(d)})$.
%

\begin{restatable}{proposition}{PropPoly}
	\label{prop-poly-init}
	Let $\A$ be a strongly connected demonic VASS with $d$ counters, and let $\bv \in \N^d$ such that $\bv(i) \leq 2^{j\cdot d}$ for every $i \leq d$, where $j < |Q|$. For every counter $c$, we have that either $\C[c,\bv] \in \Theta(n^k)$ for some $1 \leq k \leq 2^{(j{+}1)\cdot d}$, or $\C[c,\bv] \in \Omega(2^n)$. It is decidable in polynomial time which of the two possibilities holds. In the first case, the value of $k$ is computable in polynomial time. 
\end{restatable}

In~\cite{Zuleger:VASS-polynomial}, a special variant of Proposition~\ref{prop-poly-init} covering the subcase when $\bv = \vec{1}$ is proven. In the introduction part of \cite{Zuleger:VASS-polynomial}, it is mentioned that a  generalization of this result (equivalent to Proposition~\ref{prop-poly-init}) can be obtained by modifying the techniques presented in \cite{Zuleger:VASS-polynomial}. Although no explicit proof is given, the modification appears feasible. We give a simple explicit proof of Proposition~\ref{prop-poly-init}, using the algorithm of \cite{Zuleger:VASS-polynomial} for the $\bv = \vec{1}$ subcase as a ``black-box procedure''.  We refer to Appendix \ref{app-demonic} for details.

Now we extend the function $\C[c,\bv]$ so that $\bv \in \N_{\infty}^d$. Here, the $\infty$ components of $\bv$ correspond to counters that have already been pumped to ``very large'' values and do not constrain the computations in $\A$. As we shall see, ``very large'' actually means ``at least singly exponential in~$n$''.

Let $\bv \in \N_{\infty}^d$, and let $\A_{\bv}$ be the VASS obtained from $\A$ by modifying every counter update vector $\bu$ into $\bu'$, where $\bu'(i) = \bu(i)$ if $\bv(i) \neq \infty$, otherwise $\bu'(i) = 0$. Hence, the counters set to $\infty$ in $\bv$ are never changed in $\A_{\bv}$. Furthermore, let $\bv'$ be the vector obtained from $\bv$ by changing all $\infty$ components into~$1$. We put $\C_{\A}[c,\bv] = \C_{\A_{\bv}}[c,\bv']$. 

For a given $\bv \in \N_{\infty}^d$, we say that $F : \N \rightarrow \N^d$ is \emph{$\bv$-consistent} if for every $i \in \{1,\ldots,d\}$ we have that the projection $F_i : \N \rightarrow \N$ is either $\Theta(n^k)$ if $\bv_i = k$, or $2^{\Omega(n)}$ if  $\bv_i = \infty$. Intuitively, a $\bv$-consistent function assigns to every $n \in \N$ a vector $F(n)$ of initial counter values growing consistently with $\bv$.  

Given $\bv \in \N_{\infty}^d$, a control state $p \in Q$, a $\bv$-consistent function $F$, an infinite family  $\Pi = \pi_1,\pi_2,\ldots$ of Demon's strategies in $\A$, a function $S : \N \rightarrow \N$, and $n \in \N$, we use  $\beta_n[\bv,p,F,\Pi,S]$ to denote the computation of $\A$ starting at $p F(n)$ obtained by applying $\pi_n$  until a maximal computation is produced or $S(n)$ transitions are executed.


The next lemma says that if $\A$ is strongly connected, then all counters can be pumped \emph{simultaneously} to the values asymptotically equivalent to $\C_{\A}[c,\bv]$ so that the counters previously pumped to exponential values stay exponential.   

\begin{restatable}{lemma}{LemSimult}
	\label{lem-simult}
	Let $\A$ be a strongly connected demonic VASS with $d$ counters. Let $\bv \in \N_{\infty}^d$, and let $F$ be a $\bv$-consistent function. Then for every counter $c_i$ such that $\bv_i \neq \infty$ and $\C_{\A}[c_i,\bv] \in \Theta(n^k)$ we have that $\Cval[c_i](pF(n)) \in \Theta(n^k)$ for every $p \in Q$. 
	Furthermore, there exist $p \in Q$, an infinite family $\Pi$ of Demon's strategies, and a function $S \in 2^{\calO(n)}$ such that for every $c_i$, the value of $c_i$ in the last configuration of $\beta_n[\bv,p,F,\Pi,S]$ is
	\begin{itemize}
		\item $\Theta(n^k)$ if $\C_{\A}[c_i,\bv] \in \Theta(n^k)$;
		\item $2^{\Omega(n)}$ if $\bv_i = \infty$ or $\C_\A[c_i,\bv] \in 2^{\Omega(n)}$.
	\end{itemize}
\end{restatable}
\bigskip

A proof of Lemma~\ref{lem-simult} uses the result of \cite{Leroux:Polynomial-termination-VASS} saying that counters pumpable to exponential values can be simultaneously pumped by a computation of exponential length from a configuration where all counters are set to~$n$ (the same holds for polynomially bounded counters, where the length of the computation can be bounded even by a polynomial). Using the construction of  Proposition~\ref{prop-poly-init}, these results are extended to our setting with $\bv$-consistent initial counter values. Then, the initial counter values are virtually ``split into $d$~boxes'' of size $\lfloor F(n)/d\rfloor$. The computations pumping the individual counters are then run ``each in its own box'' for these smaller initial vectors and concatenated. As the computation of one ``box'' cannot affect any other ``box'', no computation can undo the effects of previous computations. The details can be found in Appendix \ref{sec-simult-proof} 


Let $\V_{\A} : \N_{\infty}^d \rightarrow \N_{\infty}^d$ be a function such that, for every $\bv \in \N_{\infty}^d$,  
\[
\V_\A(\bv)(i) = \begin{cases}
	k & \mbox{if $\bv_i \neq \infty$ and $\C_{\A}[c_i,\bv] \in \Theta(n^k)$,}\\
	\infty & \mbox{otherwise.}  
\end{cases}  
\]

Note that every SCC (vertex) $\eta$ of $\D(\A)$ can be seen as a strongly connected demonic VASS after deleting all transitions leading from/to the states outside $\eta$. If the counters are simultaneously pumped to $\bv$-consistent values before entering $\eta$, then $\eta$ can further pump the counters to $\V_{\eta}(\bv)$-consistent values (see Lemma~\ref{lem-simult}). According to Lemma~\ref{prop-poly-init}, $\V_{\eta}(\bv)$ is computable in polynomial time for every $\bv \in \N_\infty^d$ where every \emph{finite}  $\bv_i$ is bounded by $2^{j\cdot d}$ for some $j < |Q|$.




Observe that \emph{all} computations of $\A$ can be divided into finitely many pairwise disjoint classes according to their corresponding paths in $\D_\A$ (i.e., the sequence of visited SCCs of $\D_\A$).
For each such sequence $\eta_1,\ldots,\eta_m$, the vectors $\bv_0,\ldots,\bv_m$ where $\bv_0 = \vec{1}$ and $\bv_{i} = \V_{\eta_i}(\bv_{i-1})$ are computable in time polynomial in $|\A|$ (note that $m \leq |Q|$). The asymptotic growth of the counters achievable by computations following the path $\eta_1,\ldots,\eta_m$ is then given by $\bv_m$. Hence, $\C_\A[c_i] \in \Omega(n^k)$ iff there is a path $\eta_1,\ldots,\eta_m$ in $\D_\A$ such that $\bv_m(i) \geq k$. Similarly, $\C_\A[c_i] \in \calO(n^k)$ iff for every path $\eta_1,\ldots,\eta_m$ in $\D_\A$ we have that $\bv_m(i) \leq k$. From this we immediately obtain the upper complexity bounds of Theorem~\ref{thm-main-demon}.

Furthermore, for every SCC $\eta$ of $\D_\A$, we can compute the set $\Vect_\A(\eta)$ of \emph{all} $\bu$ such that there is a path $\eta_1,\ldots,\eta_m$ where $\eta_1$ is a root of $\D_\A$, $\eta_m = \eta$, and $\bu = \bv_m$. A full description of the algorithm is given in Appendix \ref{app-alg-vect}. If $\Degree(\eta)$ is bounded by a fixed constant independent of $\A$ for every leaf $\eta$ of $\D_{\A}$, then the algorithm terminates in polynomial time, which proves Theorem~\ref{thm-demon-tractable}.



\section{VASS Games}
\label{sec-VASS-games}

The computational complexity of polynomial asymptotic analysis for VASS games is classified in our next theorem. The parameter characterizing hard instances is identified at the end of this section. 
\begin{theorem}
	\label{thm-main-games}
	Let $k \geq 1$. For every VASS game $\A$ and counter \(c\) we have that $\calL$ ($\calC[c]$) is either in $\calO(n^k)$ or in $\Omega(n^{k+1})$. Furthermore, the problem whether
	\begin{itemize}
		\item $\calL \in \calO(n^k)$ is $\NP$-complete for $k {=} 1$ and $\PSPACE$-complete for $k{\geq} 2$;
		\item $\calL \in \Omega(n^{k})$ is in $\PTIME$ for $k {=} 1$, $\coNP$-complete for $k{=}2$, and $\PSPACE$-complete for $k{\geq} 3$;
		\item $\calL \in \Theta(n^{k})$ is $\NP$-complete for $k{=}1$ and $\PSPACE$-complete for $k{\geq} 2$;
		\item $\calC[c] \in \calO(n^k)$ is $\PSPACE$-complete;
		\item $\calC[c] \in \Omega(n^{k})$ is in $\PTIME$ for $k {=} 1$, and $\PSPACE$-complete for $k{\geq} 2$;
		\item $\calC[c] \in \Theta(n^{k})$ is $\PSPACE$-complete.
	\end{itemize}
\end{theorem}
Furthermore, we show that for every VASS game $\A$, either $\calL \in \calO(n^{2^{d|Q|}})$ or $\calL \in 2^{\Omega(n)}$. In the first case, the $k$ such that $\calL \in \Theta(n^k)$ can be computed in polynomial space. The same results hold for $\calC[c]$.

In \cite{KLV:VASS-Grzegorczyk}, it has been shown that the problem whether $\calL \in \calO(n)$ is $\NP$-complete, and if $\calL \not\in \calO(n)$, then $\calL \in \Omega(n^2)$. This yields the $\NP$ and $\coNP$ bounds of Theorem~\ref{thm-main-games} for $k=1,2$. Furthermore, it has been shown that the problem whether  $\calL \in \calO(n^2)$ is $\PSPACE$-hard, and this proof can be trivially generalized to obtain all $\PSPACE$ lower bounds of Theorem~\ref{thm-main-games}. For the sake of completeness, we sketch the arguments in Appendix~\ref{app-games-lower}.


The key insight behind the proof of Theorem~\ref{thm-main-games} is that player Angel can safely commit to a \emph{simple locking} strategy when minimizing the counter complexity. We start by introducing locking strategies.  

\begin{definition}
	Let $\A$ be a VASS game. We say that a strategy $\sigma$ for player Angel is \emph{locking} if for every computation $p_1\bv_1,\dots,p_m\bv_m$ where $p_m \in Q_A$ and for every $k<m$ such that $p_k = p_m$ we have that $\sigma(p_1\bv_1,\dots,p_k\bv_k) = \sigma(p_1\bv_1,\dots,p_m\bv_m)$.
\end{definition}

In other words, when an angelic control state $p$ is visited for the first time, a locking strategy selects and ``locks'' an outgoing transition of $p$ so that whenever $p$ is revisited, the previously locked transition is taken. Observe that the choice of a ``locked'' transition may depend on the whole history of a computation. 

Since a ``locked'' control state has only one outgoing transition, it can be seen as \emph{demonic}. Hence, as more and more control states are locked along a computation, the  VASS game $\A$ becomes ``more and more demonic''. We capture these changes as a finite acyclic graph $\calG_\A$ called \emph{the locking decomposition of $\A$}. Then, we say that a locking strategy is \emph{simple} if the choice of a locked transition after performing a given history depends only on the finite path in $\calG_\A$ associated to the history. We show that Angel can achieve an asymptotically optimal termination/counter complexity by using only simple locking strategies. Since the height of $\calG_\A$ is polynomial in $|\A|$, the existence of an appropriate simple locking strategy for Angel can be decided by an alternating polynomial-time algorithm. As $\AP = \PSPACE$, this proves the $\PSPACE$ upper bounds of Theorem~\ref{thm-main-games}. Furthermore, our construction identifies the structural parameters of $\calG_\A$ making the polynomial asymptotic analysis of VASS games hard. When these parameters are bounded by fixed constants, the problems of  Theorem~\ref{thm-main-games} are solvable in polynomial time. 

\subsection{Locking sets and the locking decomposition of $\A$}

Let $\A$ be a VASS game. A \emph{Demonic decomposition} of $\A$ is a finite directed graph $\D_\A$ defined as follows. Let ${\sim} \subseteq Q \times Q$ be an equivalence where $p \sim q$ iff either $p=q$, or both $p,q$ are demonic and mutually reachable from each other via a finite path leading only through demonic control states. The vertices of $\D_\A$ are the equivalence classes $Q/{\sim}$, and $[p] \rightarrow [q]$ iff $[p] \neq [q]$ and $(p,\bu,q) \in \Tran$ for some $\bu$. For demonic VASS, $\D_\A$ becomes the standard DAG decomposition. For VASS games, $\D_\A$ is \emph{not} necessarily acyclic.


A \emph{locking set} of $\A$ is a set of transitions $L \subseteq \Tran$ such that $(p,\bu,q) \in L$ implies $p \in Q_A$, and $(p,\bu,q),(p',\bu',q') \in L$ implies $p \neq p'$. A control state $p$ is \emph{locked} by $L$ if $L$ contains an outgoing transition of~$p$. We use $\cL$ to denote the set of all locking sets of~$\A$. For every $L \in \cL$, let $\A_L$ be the VASS game obtained from $\A$ by ``locking'' the transitions of $L$. That is, each control state $p$ locked by $L$ becomes demonic in $\A_L$, and the only outgoing transition of $p$ in $\A_L$ is the transition $(p,\bu,q) \in L$.

\begin{definition}
	The \emph{locking decomposition of $\A$} is a finite directed graph $\calG_\A$ where the set of vertices and the set of edges of $\calG_\A$ are the least sets $V$ and $\rightarrow$ satisfying the following conditions:
	\begin{itemize}
		\item All elements of $V$ are pairs $([p],L)$ where $L \in \cL$ and $[p]$ is a vertex of $\D_{\A_L}$. When $p$ is demonic/angelic in $\A_L$, we say that $([p],L)$ is demonic/angelic.
		\item $V$ contains all pairs of the form $([p],\emptyset)$. 	
		\item If $([p],L) \in V$ where  $p$ is demonic in $\A_L$ and $[p] \rightarrow [q]$ is an edge of $\D_{\A_L}$, then $([q],L) \in V$ and $([p],L) \rightarrow ([q],L)$.
		\item If $([p],L) \in V$ where $p$ is angelic in $\A_L$, then for every $(p,\bu,q) \in \Tran$ we have that $([q],L') \in V$ and $([p],L) \rightarrow ([q],L')$, where $L' = L \cup \{(p,\bu,q)\}$.	 
	\end{itemize}
\end{definition}

It is easy to see that $\calG_\A$ is \emph{acyclic} and the length of every path in $\calG_\A$ is bounded by~$|Q|+|Q_A|$, where at most $|Q|$ vertices in the path are demonic. Note that every computation of $\A$ obtained by applying a locking strategy determines its associated path in $\calG_\A$ in the natural way. A locking strategy $\sigma$ is \emph{simple} if the choice of a locked transition depends only on the path in $\calG_\A$ associated to the performed history. (i.e., if two histories determine the same path in $\calG_\A$, then the strategy makes the same choice for both histories).


\subsection{Upper bounds}

Let $\A$ be a VASS game with $d$ counters. For every $p \in Q$ and $\bv \in \N^d$, let $\C_\A^p[c,\bv](n)  = \Cval[c](p\bu)$ where $\bu = (n^{\bv(1)},\ldots,n^{\bv(d)})$. We extend this notation to the vectors $\bv \in \N_{\infty}^d$ in the same way as in Section~\ref{sec-upper}, i.e., for a given $\bv \in \N_{\infty}^d$, we put $\C_{\A}^p[c,\bv] = \C_{\A_{\bv}}^p[c,\bv']$. Recall that $\bv'$ is the vector obtained from $\bv$ by changing all $\infty$ components into~$1$, and $\A_{\bv}$ is the VASS obtained from $\A$ by modifying every counter update vector $\bu$ into $\bu'$, where $\bu'(i) = \bu(i)$ if $\bv(i) \neq \infty$, otherwise $\bu'(i) = 0$. 
The main technical step towards obtaining the $\PSPACE$ upper bounds of Theorem~\ref{thm-main-games} is the next proposition.

\begin{restatable}{proposition}{PropPolyGames}
	\label{prop-upper-games}
	Let $\A$ be a VASS game with $d$ counters. Furthermore, let $([p],L)$ be a vertex of $\calG_\A$, $\bv \in \N_{\infty}^d$, and $c_i$ a counter such that $\bv_i \neq \infty$. Then, one of the following two possibilities holds:
	\begin{itemize}
		\item there is $k \in \N$ such that 
		for every $\bv$-consistent $F$ there exist a simple locking Angel's strategy $\sigma_\bv$ in $\A_L$ and a Demon's strategy $\pi_\bv$ in $\A_L$ such that $\sigma_\bv$ is independent of $F$ and
		\begin{itemize}
			\item for every Demon's strategy $\pi$ in $\A_L$, we have that   $\mmax[c_i](\Comp_{\A_L}^{\sigma_\bv,\pi}(p\,F(n))) \in \calO(n^k)$;
			\item for every Angel's strategy $\sigma$ in $\A_L$, we have that
			$\mmax[c_i](\Comp_{\A_L}^{\sigma,\pi_\bv}(p\,F(n))) \in \Omega(n^k)$.  
		\end{itemize}
		\item 
		for every  $\bv$-consistent $F$ there is a Demon's strategy $\pi_\bv$ in $\A_L$ such that for every Angel's strategy $\sigma$ in $\A_L$, we have that 
		$\mmax[c_i](\Comp_{\A_L}^{\sigma,\pi_\bv}(p\,F(n))) \in 2^{\Omega(n)}$.
	\end{itemize}
\end{restatable}

Proposition~\ref{prop-upper-games} is proven by induction on the height of the subgraph rooted by 
$([p],L)$. The case when $([p],L)$ is demonic (which includes the base case when $([p],L)$ is a leaf) follows from the constructions used in the proof of Proposition~\ref{prop-poly-init}. When the vertex $([p],L)$ is angelic, it has immediate successors of the form $([q_i],L_i)$ where $L_i=L \cup \{(p,\bu_i,q_i)\}$. We show that by locking one of the $(p,\bu_i,q_i)$ transitions in $p$, Angel can minimize the growth of $c_i$ in asymptotically the same way as if he used all of these transitions freely when revisiting~$p$. See Appendix \ref{app-games-upper} for details. 

Observe that every computation in $\A$ where Angel uses some simple locking strategy determines the unique corresponding path in $\calG_\A$ (initiated in a vertex of the form $([p],\emptyset)$) in the natural way. Hence, all such computations can be divided into finitely many pairwise disjoint classes according to their corresponding paths in $\calG_\A$. Let $([p_1],L_1),\ldots,([p_k],L_k)$ be a path in $\calG_\A$ where $L_1 = \emptyset$. Consider the corresponding sequence $\bv_0,\ldots,\bv_k$ where $\bv_0 = \vec{1}$ and $\bv_i$ is equal either to $\V_{[p_i]}(\bv_{i-1})$ or to $\bv_{i-1}$, depending on whether $([p_i],L_i)$ is demonic or angelic, respectively. Here, $\V$ is the function defined in Section~\ref{sec-upper} (observe that the component $[p]$ of $\D_{A_L}$ containing $p$ can be seen as a strongly connected demonic VASS after deleting all transitions from/to the states outside $[p]$). The vector $\bv_k$ describes the maximal asymptotic growth of the counters achievable by Demon when Angel uses the simple locking strategy associated to the path. Furthermore, the sequence $\bv_0,\ldots,\bv_k$ is computable in time polynomial in $|\A|$ and all finite components of $\bv_k$ are bounded by $2^{d\cdot |Q|}$ because the total number of all demonic $([p_i],L_i)$ in the path is bounded by $|Q|$ (cf.{} Proposition~\ref{prop-poly-init}).

The problem whether $\calC[c_i] \in \calO(n^k)$ can be decided by an \emph{alternating} polynomial-time algorithm which selects an initial vertex of the form $([p],\emptyset)$ universally, and then constructs a maximal path in $\calG_\A$ from $([p],\emptyset)$ where the successors of demonic/angelic vertices are chosen universally/existentially, respectively. After obtaining a maximal path $([p_1],L_1),\ldots,([p_k],L_k)$, the vector $\bv_k$ is computed in polynomial time, and the algorithm answers yes/no depending on whether $\bv_k(i) \leq k$ or not, respectively. The problem  whether $\calC[c_i] \in \Omega(n^k)$ is decided similarly, but here the initial vertex is chosen existentially, the successors of demonic/angelic vertices are chosen existentially/universally, and the algorithm answers yes/no depending on whether $\bv_k(i) \geq k$ or not, respectively. This proves the $\PSPACE$ upper bounds of Theorem~\ref{thm-main-games}.

Observe that the crucial parameter influencing the computational hardness of the asymptotic analysis for VASS games is the number of maximal paths in $\calG_\A$. If $|Q_A|$ and $\Degree([p],L)$ are bounded by constants, then the above alternating polynomial time algorithms can be simulated by \emph{deterministic} polynomial time algorithms. Thus, we obtain the following:

\begin{theorem}
	\label{thm-games-tractable}
	Let $\Lambda$ be a class of VASS games such that for every $\A \in \Lambda$ we have that $|Q_A|$ and $\Degree([p],L)$, where $([p],L)$ is a leaf of $\calG_\A$, are bounded by a fixed constant depending only on $\Lambda$. 
	Then, the problems whether $\calL_\A \in \calO(n^k)$, $\calL_\A \in \Omega(n^{k})$, $\calL_\A \in \Theta(n^{k})$ for given $\A \in \Lambda$ and $k \in \N$, are solvable in polynomial time (where the $k$ is written in binary). The same results hold also for $\calC[c]$ (for a given counter $c$ of $\A$).
\end{theorem}   


\section{Conclusions, future work}
\label{sec-concl}

We presented a precise complexity classification for the problems of polynomial asymptotic complexity of demonic VASS and VASS games. We also identified the structural parameters making these problems computationally hard, and we indicated that these parameters may actually stay reasonably small when dealing with VASS abstractions of computer programs. The actual applicability and scalability of the presented results to the problems of program analysis requires a more detailed study including experimental evaluation.

From a theoretical point of view, a natural question is whether the scope of effective asymptotic analysis can be extended from purely non-deterministic VASS to VASS with probabilistic transitions (i.e., VASS MDPs and VASS stochastic games). These problems are challenging and motivated by their applicability to probabilistic program analysis.


\appendix\newpage
\section{The Existence of Equilibrium Value in VASS Games}
\label{app-determinacy}

In this section, we sketch a proof for the equalities
\begin{eqnarray*}
	\sup_{\pi} \ \Inf_{\sigma} \ \len(\Comp^{\sigma,\pi}(p\bv)) & = & \Inf_{\sigma} \ \sup_{\pi}\   \len(\Comp^{\sigma,\pi}(p\bv))\\[2ex]
	\sup_{\pi} \ \Inf_{\sigma} \ \mmax[c](\Comp^{\sigma,\pi}(p\bv)) & = & \Inf_{\sigma} \ \sup_{\pi}\   \mmax[c](\Comp^{\sigma,\pi}(p\bv))
\end{eqnarray*} 
used in Section~\ref{sec-prelim}. These results appear folklore and plausible. Still, they do not immediately follow from the standard determinacy results for Borel objectives because the reward function is not bounded. Due to the importance of these equalities, we believe they are worth a proof sketch. 

We prove the determinacy result for arbitrary finitely-branching games with countably many vertices. 
Consider a game $G = (V,{\rightarrow},c)$ where $V$ is a countably infinite set of vertices partitioned into the subsets $V_D$ and $V_A$ of demonic and angelic vertices,  ${\rightarrow} \subseteq V \times V$ is a finitely-branching transition relation (i.e., every vertex $v$ has only finitely many immediate successors), and $c : V \rightarrow \N$ is a cost function. Furthermore, we fix an initial vertex $\hat{v}$. For every maximal path $\alpha$ initiated in $\hat{v}$, let $\mmax[c](\alpha)$ be the supremum of the costs of all vertices visited by $\alpha$. We show that 
\begin{equation*}
	\sup_{\pi} \ \Inf_{\sigma} \ \mmax[c](\alpha^{\sigma,\pi}) \quad = \quad \Inf_{\sigma} \ \sup_{\pi}\   \mmax[c](\alpha^{\sigma,\pi})
\end{equation*}
where $\pi$ and $\sigma$ range over the strategies for Demon/Angel in $G$, and $\alpha^{\sigma,\pi}$ is the unique maximal path initiated in $\hat{v}$ determined by $\pi$ and $\sigma$.

Let $\Gamma : \N_{\infty}^V \rightarrow \N_{\infty}^V$ be a (Bellman) operator such that, for a given $f : V \rightarrow \N_\infty$, we have that $\Gamma(f) = g$, where $g : V \rightarrow \N_{\infty}$ is defined as follows:
\[
g(v) = \begin{cases}
	\max \big\{ c(v), \max\{f(u) \mid v \rightarrow u \} \big\}    & \mbox{if } v\in V_D\\
	\max \big\{ c(v), \min\{f(u) \mid v \rightarrow u \} \big\}   & \mbox{if } v\in V_A
\end{cases}
\]

Since $\Gamma$ is a continuous operator over the CPO of all functions $V \rightarrow \N_\infty$ with component-wise ordering, there is the least fixed-point $\mu\Gamma = \bigsqcup_{i=0}^{\infty} \Gamma^i(\perp)$ of $\Gamma$, where $\perp$ is the least element (i.e., a function assigning $0$ to every vertex). Consider two memoryless strategies $\pi^*$ and $\sigma^*$ for Demon and Angel such that
\begin{itemize}
	\item for every $v \in V_D$, the strategy $\pi^*$ selects a successor of $v$ with the maximal $\mu\Gamma$ value;
	\item for every $v \in V_A$, the strategy $\sigma^*$ selects a successor of $v$ with the minimal $\mu\Gamma$ value. 
\end{itemize}
Now, it suffices to show that
\begin{equation}
	\mu\Gamma(\hat{v}) \quad \leq \quad \sup_{\pi} \ \Inf_{\sigma} \ \mmax[c](\alpha^{\sigma,\pi}) \quad \leq  \quad \Inf_{\sigma} \ \sup_{\pi}\  \mmax[c](\alpha^{\sigma,\pi})  \quad \leq \quad \mu\Gamma(\hat{v})
	\label{ineq-value}
\end{equation}
The second inequality of~\eqref{ineq-value} holds trivially. For the first inequality, observe that 
\[
\Inf_{\sigma} \ \mmax[c](\alpha^{\sigma,\pi^*}) \quad \leq \quad \sup_{\pi} \ \Inf_{\sigma} \ \mmax[c](\alpha^{\sigma,\pi})
\]
Hence, it suffices to show $\mu\Gamma(\hat{v}) \leq \Inf_{\sigma}  \mmax[c](\alpha^{\sigma,\pi^*})$, which is achieved by demonstrating $\Gamma^i(\perp)(\hat{v}) \leq \Inf_{\sigma}  \mmax[c](\alpha^{\sigma,\pi^*})$ for every $i \in \N$ (by induction on~$i$). The last inequality in~\eqref{ineq-value} is proven similarly (using $\sigma^*$).

\section{Proofs for Section~\ref{sec-demonic}}
\label{app-demonic}

\subsection{A proof of Lemma~\ref{lem-hardness2}}

\LemDPhard*
\begin{proof}
	The structure of $\A_{\varphi,\psi}$ is given by the VASS Program~\ref{alg-Aphipsi}. The program starts by executing the first eight lines of the VASS Program~\ref{alg-Aphi} constructed for $\varphi$ and $k-1$, followed by the first eight lines of the same program constructed for $\psi$ and $k-1$ where all variables representing the counters are fresh and previously unused. According to the proof of Lemma~\ref{lem-hardness1}, the Demon can perform $\Theta(n^{k-1})$ transitions when executing the first two lines of the VASS Program~\ref{alg-Aphipsi} regardless whether the formulae $\varphi,\psi$ are satisfiable or not. Let $s(\A_\varphi)$ and $s(\A_\psi)$ be the counters of $\A_\varphi$ and $\A_\psi$ corresponding to the counter $s_m$ of the VASS Program~\ref{alg-Aphi}. According to Lemma~\ref{lem-hardness1}, we have the following:
	\begin{itemize}
		\item If $\varphi$ is satisfiable, then the counter $s(\A_\varphi)$ can be pumped to $\Theta(n^{k-1})$; otherwise, it can be pumped only to $\Theta(n)$.
		\item If $\psi$ is satisfiable, then the counter $s(\A_\psi)$ can be pumped to $\Theta(n^{k-1})$; otherwise, it can be pumped only to $\Theta(n)$.
	\end{itemize}
	The instructions at lines~\ref{line-multi} and~\ref{line-square} ensure that $\calL_{\A_{\varphi,\psi}} \in \Theta(n^k)$ iff the counter $s(\A_\varphi)$ can be pumped to $\Theta(n^{k-1})$ and the counter $s(\A_\psi)$ can be pumped only to $\Theta(n)$. More precisely, the instructions at line~\ref{line-multi} multiply the values of these two counters. Hence, if both values are $\Theta(n^{k-1})$, the multiplication gadget executes $\Theta(n^{2k-2})$ transitions, which is beyond $\Theta(n^{k})$. If one of the counter values is $\Theta(n)$ and the other is $\Theta(n^{k-1})$, the multiplication takes $\Theta(n^{k})$ transitions. Finally, if both values are $\Theta(n)$, the multiplication takes $\Theta(n^{2})$ transitions. The instructions at line~\ref{line-square} compute the square of the value stored in $s(\A_\psi)$, which takes either $\Theta(n^{2k-2})$ or $\Theta(n^{2})$ transitions, depending on whether the value of $s(\A_\psi)$ is $\Theta(n^{k-1})$ or $\Theta(n)$, respectively. So, the only case when these instructions produce a sequence of transitions of length $\Theta(n^{k})$ is when the value of $s(\A_\varphi)$ is $\Theta(n^{k-1})$ and the value of  $s(\A_\psi)$ is $\Theta(n)$.
\end{proof}

\SetAlgorithmName{VASS Program}{vassprog}{VASS Programs}
\SetAlFnt{\small}
\SetInd{0ex}{1.55ex}
\begin{algorithm}[tb]
	\SetAlgoLined
	\DontPrintSemicolon
	\SetKwInOut{Parameter}{parameter}\SetKwInOut{Input}{input}\SetKwInOut{Output}{output}
	\SetKwData{C}{C}\SetKwData{D}{D}\SetKwData{MX}{M}\SetKwData{f}{f}
	\SetKw{Choose}{choose:}	
	\SetKw{Or}{or}
	\SetKw{And}{and}
	\BlankLine
	lines~\ref{dk-init}--\ref{line-end-loop} of $\A_\varphi$ constructed for $k-1$; \;
	lines~\ref{dk-init}--\ref{line-end-loop} of $\A_\psi$ constructed for $k-1$; \quad/* all counters are fresh */ \;
	$a \leftarrow s(\A_\varphi);$ \ \ $b \leftarrow s(\A_\psi);$ \ \ $e \leftarrow a*b$; \quad\quad/* $a,b,e$ are fresh counters */\; \label{line-multi}
	$c \leftarrow s(\A_\psi);$ \ \ $d \leftarrow s(\A_\psi);$ \ \ $f \leftarrow c*d$; \ \, \quad/* $c,d,f$ are fresh counters */ \label{line-square}
	\caption{$\A_{\varphi,\psi}$}
	\label{alg-Aphipsi}
\end{algorithm}

\subsection{A proof of Lemma~\ref{lem-hardness3}}

\LemDPhardcounter*
\begin{proof}
	The structure of $\A_{\varphi,\psi}$ is given by the VASS Program~\ref{alg-AphipsiC}. The program starts by executing the first eight lines of the VASS Program~\ref{alg-Aphi} constructed for $\varphi$ and $k+1$, followed by the first eight lines of the same program constructed for $\psi$ and $k+1$ where all variables representing the counters are fresh and previously unused. Let $s(\A_\varphi)$ and $s(\A_\psi)$ be the counters of $\A_\varphi$ and $\A_\psi$ corresponding to the counter $s_m$ of the VASS Program~\ref{alg-Aphi}. According to Lemma~\ref{lem-hardness1}, we have the following:
	\begin{itemize}
		\item If $\varphi$ is satisfiable, then the counter $s(\A_\varphi)$ can be pumped to $\Theta(n^{k+1})$; otherwise, it can be pumped only to $\Theta(n)$.
		\item If $\psi$ is satisfiable, then the counter $s(\A_\psi)$ can be pumped to $\Theta(n^{k+1})$; otherwise, it can be pumped only to $\Theta(n)$.
	\end{itemize}
	The instructions at line~\ref{dk-init2} pump \(d_{k-1} \) to \(\Theta(n^{k-1}) \).	Thus the instructions at line~\ref{line-multiC} ensure if \(\varphi \) is unsatisfiable then \(c \) can be pumped at most to \(\Theta(n) \), and if \(\varphi \) is satisfiable then \(c \) can be pumped at most to \(\Theta(n^{k+1}) \) or \(\Theta(n^{k}) \) depending on whether \( \psi\) is satisfiable or unsatisfiable, respectively. Therefore $\calC_{\A_{\varphi,\psi}}[c] \in \Theta(n^k)$ iff \(\varphi \) is satisfiable and \(\psi \) is unsatisfiable.
\end{proof}

\SetAlgorithmName{VASS Program}{vassprog}{VASS Programs}
\SetAlFnt{\small}
\SetInd{0ex}{1.55ex}
\begin{algorithm}[tb]
	\SetAlgoLined
	\DontPrintSemicolon
	\SetKwInOut{Parameter}{parameter}\SetKwInOut{Input}{input}\SetKwInOut{Output}{output}
	\SetKwData{C}{C}\SetKwData{D}{D}\SetKwData{MX}{M}\SetKwData{f}{f}
	\SetKw{Choose}{choose:}	
	\SetKw{Or}{or}
	\SetKw{And}{and}
	\BlankLine
	lines~\ref{dk-init}--\ref{line-end-loop} of $\A_\varphi$ constructed for $k+1$; \;
	lines~\ref{dk-init}--\ref{line-end-loop} of $\A_\psi$ constructed for $k+1$; \quad/* all counters are fresh */ \;
	$d_2 \leftarrow d_1 * e_1;\ \ d_3 \leftarrow d_2 * e_2; \ \cdots \ ; \ \ d_{k-1} \leftarrow d_{k-2} * e_{k-2};$\label{dk-init2} \quad/* all counters are fresh */ \;
	$c \leftarrow s(\A_\varphi);$ \ \ $b \leftarrow s(\A_\psi) * d_{k-1};$ \ \ $c \leftarrow min(a,b)$; \quad/* $a,b,c$ are fresh counters */\; \label{line-multiC}
	\caption{$\A_{\varphi,\psi}$}
	\label{alg-AphipsiC}
\end{algorithm}

\subsection{A proof of Proposition~\ref{prop-poly-init}}

In this section we give a full proof of Proposition~\ref{prop-poly-init}. We start by recalling the special variant proven in \cite{Zuleger:VASS-polynomial}.

\begin{proposition}[see \cite{Zuleger:VASS-polynomial}]
	\label{prop-zuleger}
	Let $\A$ be a strongly connected demonic VASS with $d$ counters. For every counter $c$, we have that either $\C[c] \in \Theta(n^k)$ for some $1 \leq k \leq 2^d$, or $\C[c] \in \Omega(2^n)$. It is decidable in polynomial time which of the two possibilities holds. In the first case, the $k$ is computable in polynomial time. 
\end{proposition}

\noindent
Our proof of Proposition~\ref{prop-poly-init} is obtained by modifying a given demonic VASS $\A$ into another demonic VASS $\hat{\A}$ and applying Proposition~\ref{prop-zuleger} to $\hat{\A}$. 
\PropPoly*
\begin{proof}
	Let $c_1,\ldots,c_d$ be the counters of $\A$. We start by constructing a VASS $\B$ that ``pumps'' every $c_i$ to $n^{\bv(i)}$ from an initial configuration with all counters set to $n$ (see VASS Program~\ref{alg-pumping}).  Since the components of $\bv$ can be exponential in $d$, the 
	trivial technique used in line~\ref{dk-init} of VASS Program~\ref{alg-Aphi} is not applicable because it requires $\Omega(2^{j\cdot d})$ new counters and control states. Hence, $\B$ needs to be constructed more carefully using repeated squaring. Let $\max_\bv$ be the maximal component of $\bv$, and let $\ell = \lfloor \log(\max_\bv) \rfloor$. Then, for every $1 \leq i \leq d$, there is a vector $\vec{t}_i \in \{0,1\}^{\ell+1}$ computable in time polynomial in $|\A|$ such that $\bv(i) = \vec{t}_i * (2^0,\ldots,2^{\ell})$. Let $f_i$ be the least index $k$ such that $\vec{t}_i(k)=1$. At line~\ref{line-mult}, the counter $m_j$ is pumped to $\Theta(n^{2^j})$.
	Note that when maximizing the value of $m_j$, the corresponding gadget (see Fig.~\ref{fig-gadgets}) leaves $0$ either in $m_{j-1}$ or in $\alpha$, and the value of both counters is at most~$1$. The nested loop at lines~\ref{line-pump-start}--\ref{line-pump-end} increase the counter $s_i$ by $n^{\vec{t}_i * (2^0,\ldots,2^{\ell})}$ for every $i = 1,\ldots,d$ using the $m_j$ counters. Note that every counter is initialized to $n$, so no computation is needed for $j=0$. Also note that the \textbf{if} statements at  
	lines~\ref{line-if1}--\ref{line-if3} are purely symbolic---the conditions $\vec{t}_i(j) = 1$, $f_i=j$ and $f_i\neq j$ are constants denoting either true or false, and the instructions following \textbf{then} are either included into the code of $\B$ or not. At lines~\ref{line-copy-start}--\ref{line-copy-end}, the values stored in $s_1,\ldots,s_d$ are added to the counters $c_1,\ldots,c_d$. The gadget for the instruction $x \leftarrow [y]$ ``reads'' the value of $y$ destructively, i.e., the counter $y$ is \emph{not} restored to its original value (cf.\ the gadget for  $x \leftarrow y$ of Fig.~\ref{fig-gadgets}).

	Note that the VASS $\B$ has two distinguished control states $\textit{in}$ and $\textit{out}$, and uses $\kappa$ auxiliary counters different from $c_1,\ldots,c_d$. Hence, the total number of counters of $\B$ is $d+\kappa$. The transition update vectors of $\B$ are constructed so that the first $d$ components specify the updates for $c_1,\ldots,c_d$. An important observation is that even if we extended $\B$ with a transition $(\textit{out},\vec{0},\textit{in})$ allowing for ``restarting'' the computation of $\B$, this extra transition would be of no use when maximizing the values of $c_1,\ldots,c_\ell$. The best the Demon could do is to maximize the value of all $m_j$, then maximize all $s_i$, and then empty $s_i$ by transferring its content to~$c_i$. If the Demon violated from this scenario, leaving some positive values in the auxiliary counters of $\B$ and then ``restarting'' the computation of $\B$ using the transition $(\textit{out},\vec{0},\textit{in})$, the resulting value of $c_1,\ldots,c_\ell$ would be only smaller.
	
	Now we construct a VASS $\U$ with $d+\kappa$ counters by taking the union of $\A$ and $\B$, where the transition update vectors of $\A$ are extended so that they do not modify the extra $\kappa$ counters of $\B$. Furthermore, for every control state $p$ of $\A$, we add to $\U$ the transitions $(p,-\vec{m},\textit{in})$ and $(\textit{out},-\vec{m},p)$. Here, $m = |Q|\cdot M$, where $M$ is the maximal absolute value of an update vector component in $\A$. Observe that $\U$ is strongly connected and its size is polynomial in the size of $\A$. 
	
	We show that for every $c_i$, where $i \in \{1,\ldots,d\}$, the asymptotic growth of $\C[c_i]$ in $\U$ is the same as the asymptotic growth of $\C[c_i,\bv]$ in $\A$. Hence, it suffices to apply  Proposition~\ref{prop-zuleger} to $\U$.
	
	Clearly, $\C_{\A}[c_i,\bv] \in \calO(\C_{\U}[c_i])$ because $\U$ can use the sub-VASS $\B$ to pump every $c_j$ to $\Theta(n^{\bv(j)})$ and then simulate a computation of $\A$. It remains to show $\C_{\U}[c_i] \in \calO(\C_{\A}[c_i,\bv])$. To see this, it suffices to verify that the best strategy for the Demon who aims at maximizing the value of $c_i$ in a computation of $\U$ initiated in a configuration with all counters set to $n$ is to start in the $\textit{in}$ state of the sub-VASS~$\B$ and pump all $c_1,\ldots,c_d$ to their maximal values, and then continue by simulating $\A$ without ever returning to the $\textit{in}$ state of the sub-VASS~$\B$. Such a computation can be ``simulated'' by $\A$ from an initial configuration where every $c_j$ is set to $n^{\bv(j)}$, which proves our claim.
	
	Obviously, an extra ``detour'' to $\B$ is of no use if the counters $c_1,\ldots,c_d$ have previously been pumped to their maximal values. As noted above, the Demon cannot gain anything by deviating from the scenario when $c_1,\ldots,c_d$ are pumped just once at the beginning because the total increase in the values of $c_1,\ldots,c_d$ brought by running the sub-VASS $\B$ could be only smaller. So, the only reason why Demon might still wish to re-visit $\B$ by entering \textit{in} from a control state $p$ of the sub-VASS $\A$ is a ``shorter'' path to some other control state $q$ of the sub-VASS $\A$ passing through the sub-VASS~$\B$. However, since $\A$ is strongly connected, the Demon can always move from $p$ to $q$ via the control states of $\A$ and the total decrease in every counter will be only smaller than the decrease caused by the transitions  $(p,-\vec{m},\textit{in})$ and $(\textit{out},-\vec{m},q)$. To sum up, the optimal strategy for Demon is to use the sub-VASS $\B$ only at the beginning to pump all $c_1,\ldots,c_d$ to their maximal values, and then schedule an appropriate computation of the sub-VASS~$\A$. 
\end{proof}

\SetAlgorithmName{VASS Program}{vassprog}{VASS Programs}
\SetAlFnt{\small}
\SetInd{0ex}{3ex}
\begin{algorithm}[tb]
	\SetAlgoLined
	\DontPrintSemicolon
	\SetKwInOut{Parameter}{parameter}\SetKwInOut{Input}{input}\SetKwInOut{Output}{output}
	\SetKwData{C}{C}\SetKwData{D}{D}\SetKwData{MX}{M}\SetKwData{f}{f}
	\SetKw{Choose}{choose:}	
	\SetKw{Where}{where}	
	\SetKw{Or}{or}
	\SetKw{And}{and}
	\BlankLine
	\ForEach{$j = 1,\ldots,\ell$\label{compute-m-start}}{
		$m_j \leftarrow m_{j-1} * m_{j-1};$ \;\label{line-mult}
		\ForEach{$i = 1,\ldots,d$\label{line-pump-start}}{
			\If{$\vec{t}_i(j) = 1$\label{line-if1}}{
				\lIf{$f_i = j$\label{line-if2}}{
					$s_i \leftarrow m_j;$
				}
				\lIf{$f_i \neq j$\label{line-if3}}{
					$a(i,j) \leftarrow s_i;\ \ s_i 	\leftarrow m_j * a(i,j);$
				}
			}
		}\label{line-pump-end}
	}	 
	
	\ForEach{$i = 1,\ldots,d$\label{line-copy-start}}{
		$c_i \leftarrow [s_i];$
		\smash{\hspace*{25em}%
			\begin{tikzpicture}[scale=1, x=2cm, y=2cm, font=\footnotesize]
				\node[stoch] (n1) at (0,0) {in};
				\node[stoch] (n2) at (1,0) {out};
				\draw [tran,loop above] (n1) to node[above] {\parbox{3em}{\centering -- $y$\\+ $x$}} (n1); 
				\draw [tran] (n1) -- (n2);
				\node at (.5,-.4) {$x \leftarrow [y]$};
			\end{tikzpicture}	  
		}	
	}\label{line-copy-end}			
	\caption{The ``pumping'' VASS $\B$}
	\label{alg-pumping}
\end{algorithm}

\subsection{A Proof of Lemma~\ref{lem-simult}}
\label{sec-simult-proof}

First, let us restate Lemma~\ref{lem-simult}.

\LemSimult*
\bigskip 

Let $\A$ be a strongly connected VASS. For every cycle $\gamma$ of $\A$, let $\eff{\gamma}$ be the sum of the counter update vectors of the transitions executed along $\gamma$. 

First we will need the following result of \cite{Leroux:Polynomial-termination-VASS}. Let $\calE$ be the set of all counters $c$ such that $\C_{\A}[c] \in 2^{\Omega(n)}$. Then there exists an \emph{iteration scheme for $\calE$}, i.e., a sequence of cycles $\gamma_1,\ldots,\gamma_k$ such that every counter strictly decremented by some $\eff{\gamma_i}$ is strictly incremented by  $\sum_{i=1}^k \eff{\gamma_i}$. Furthermore, $\calE$ is precisely the set of all counters strictly incremented by $\sum_{i=1}^k \eff{\gamma_i}$. In \cite{Leroux:Polynomial-termination-VASS}, it is shown that an iteration scheme can be ``iterated exponentially many times'', producing a computation of length $2^{\calO(n)}$ such that all counters of $\calE$ are simultaneously pumped to the value $2^{\Omega(n)}$ (see Lemma~10 in \cite{Leroux:Polynomial-termination-VASS}).

Second, every counter $c$ such that $\C_{\A}[c] \in \calO(n^k)$ can be pumped to the value $\Omega(n^k)$ by a computation of polynomial length (using the results of \cite{Zuleger:VASS-polynomial}, it is easy to show that the length of the computation can be bounded by $\calO(n^{k+1})$; for our purposes, even a singly exponential bound is sufficient, so there is no need to go into the details). 


The above mentioned results assume that the vector of initial counter values is $\vec{n}$. Now let $\bv \in \N_{\infty}^d$.
Consider a VASS $\hat{\A}_{\bv}$ obtained by first modifying $\A$ into $\A_\bv$ (i.e., for every $c_i$ such that $\bv_i = \infty$, every counter update vector $\bu$ of $\A$ is modified so that $\bu_i = 0$, see Section~\ref{sec-upper}), and then extending $\A_\bv$ into $\hat{\A}_{\bv}$ by the construction of Proposition~\ref{prop-poly-init}. That is, $\hat{\A}_{\bv}$ is obtained from $\A_\bv$ by adding the gadget   pumping every counter $c_i$ such the $\bv_i = k <\infty$ to $\Theta(n^k)$ from an initial configuration where all counters are set to~$n$. Now, the above results are applicable to $\hat{\A}_{\bv}$. That is, for every counter $c_i$ such that $\bv_i \neq \infty$, there exist 

\begin{itemize}
	\item a control state $p_i \in Q$,
	\item an infinite family  $\tau^i_1,\tau^i_2,\ldots$ of strategies for player Demon in $\hat{\A}_{\bv}$,
	\item a function $S_i : \N \rightarrow \N$ such that $S_i \in 2^{\calO(n)}$, 
\end{itemize}

\noindent
such that for every $n \in \N$, the value of $c_i$ in the configuration reached by applying $\tau^i_n$ 
from $p_i \vec{n}$ until a maximal computation is obtained or $S_i(n)$ transitions are executed, is 

\begin{itemize}
	\item $\Theta(n^k)$ if $\C_{\hat{A}_\bv}[c] \in \Theta(n^k)$,
	\item $2^{\Omega(n)}$ if $\C_{\hat{A}_\bv}[c] \in 2^{\Omega(n)}$.
\end{itemize}

By applying the observations of Proposition~\ref{prop-poly-init}, we can safely assume that $p_i = \textit{in}$ where $\textit{in}$ is the distinguished starting state of the pumping gadget, and the above computations never revisit the control state $\textit{in}$ after passing through the control state $\textit{out}$ of the gadget. Let $P_i(n)$ be the length of the computation from configuration $\textit{in}\thinspace\vec{n}$ under strategy $\tau^i_n$ until the control state $\textit{out}$ is reached (and $\infty$ if $\textit{out}$ is never reached), and let $T_i : \N \rightarrow \N$ be such that $T_i(n)=\max (0,S_i(n)-P_i(n)-1)$. Note that $T_i\in 2^{\calO(n)}$.


For every $i \in \{1,\ldots,d\}$, let $\Gad_i(n)$ be the maximal value of $c_i$ achievable by running the pumping gadget of $\hat{A}_\bv$ from the control state $\textit{in}$ where all counters are initialized to $n$ (cf.~VASS~Program~\ref{alg-pumping}). Recall that $\Gad_i(n) \in \Theta(n^{\bv_i})$. Also let $\Gad(n)=(\Gad_1(n),\dots,\Gad_d(n))$. 

Let $F$ be a $\bv$-consistent function. We define a function $R : \N \rightarrow \N$ such that $R(n)$ is the largest $m$ satisfying the following condition:

\begin{itemize}
	\item $\lfloor F_i(n)/d \rfloor \geq \Gad_i(m) + |Q| \cdot M$ for every $i \in \{1,\ldots,d\}$, where $Q$ are the control states of $\A$ and $M$ is the maximal absolute value of an update vector component of $\A$;
\end{itemize}

\noindent
Observe that $R \in \Theta(n)$. For every $n \in \N$ and $i \in \{1,\ldots,d\}$ such that $\bv_i \neq \infty$, consider a Demon's strategy $\varrho^i_n$ in $\A$ defined as follows: Let $q$ be the control state visited by the strategy $\tau^i_{R(n)}$ right after leaving the control state $\textit{out}$ of the pumping gadget of $\hat{A}_\bv$. The strategy $\varrho^i_n$ takes the shortest path to $q$ and then starts to behave exactly like $\tau^i_{R(n)}$. Note that $\varrho^i_n$ can faithfully simulate $\tau^i_{R(n)}$ either until $\tau^i_{R(n)}$ produces a maximal computation (for initial configuration $q\thinspace\Gad(R(n))$) or for at least $T_i(R(n))$ steps, which is sufficient to pump the counter $c_i$ to $\Theta(n^k)$ if $\C_\bv[c_i] \in \Theta(n^k)$, or to $2^{\Omega(n)}$ if $\C_\bv[c_i] \in 2^{\Omega(n)}$. The strategy $\pi_n$ is obtained by ``concatenating'' all $\varrho^i_n$ (for all counters $c_i$ such that $\bv_i \neq \infty$), where $\varrho^i_n$ is executed either for $T_i(R(n))$ steps since initiating the simulation of $\tau^i_{R(n)}$ or until $\tau^i_{R(n)}$ produces a maximal computation. Note that each simulation of   $\tau^i_{R(n)}$ "assumes" that the counters are set to $\Gad(R(n))$, and therefore the most any $\varrho^i_n$ can lower any counter $c_j$ is by $\Gad_j(R(n))+|Q| \cdot M$. Therefore we can assume that $\pi_n$ "splits the counters into $d$ boxes" of size $\lfloor F(n)/d \rfloor$, and performs each simulation in it's own box. Hence, no simulation can "undo" the effect of the previous simulations. Also let $S(n)$ be the maximal number of steps $\pi_n$ takes until the last simulation $\varrho^i_n$ is finished from initial configuration $pF(n)$ for any $p\in Q$. The total length of the whole simulation is bounded by $S(n) \leq \sum_{i, \bv_i \neq \infty} \left( |Q| + T_i(R(n) \right)$. Observe that $S(n) \in 2^{\calO(n)}$, and for every counter $c_i$ such that $\bv_i = \infty$ we have that the value of $c_i$ after performing $\pi_n$ in the above indicated way is at least  $\lfloor F_i(n)/d\rfloor$, which is $2^{\Omega(n)}$, since the $i$-th box "remains untouched".


\subsection{An Algorithm Computing $\Vect_{\A}$}
\label{app-alg-vect}

In this section, we present an algorithm computing the set $\Vect_{\A}(\eta)$ for every SCC $\eta$ of $\D_\A$. Recall that $\Vect_\A(\eta)$ consists of \emph{all} $\bu$ such that there is a path $\eta_1,\ldots,\eta_m$ where $\eta_1$ is a root of $\D_\A$, $\eta_m = \eta$, and $\bu = \bv_m$. 

The sets $\Vect_{\A}(\eta)$ are computed by Algorithm~\ref{alg-vectors}. In particular, at lines~\ref{root-start}--\ref{root-end},  $\Vect_\A(\eta)$ is set to $\V_{\eta}(\vec{1})$ for every root $\eta$ of $\D(\A)$. The algorithm then follows the top-down acyclic structure of $\D(\A)$ and computes $\Vect_\A(\eta)$ for the remaining components. 

\SetAlgorithmName{Algorithm}{algo}{Algorithms}
\SetAlFnt{\small}
\SetInd{0ex}{3ex}
\begin{algorithm}[tb]
	\SetAlgoLined
	\DontPrintSemicolon
	\SetKwInOut{Parameter}{parameter}\SetKwInOut{Input}{input}\SetKwInOut{Output}{output}
	\SetKwData{C}{C}\SetKwData{D}{D}\SetKwData{MX}{M}\SetKwData{f}{f}
	\SetKw{Or}{or}
	\SetKw{And}{and}
	\SetKw{Where}{where}
	\Input{A Demonic decomposition $\D(\A)$ of a demonic VASS $\A$} 
	\Output{The function $\Vect_\A$}
	\BlankLine
	\ForEach{vertex $\mu$ of $\D(\A)$\label{empty-start}}{
		$\Vect_\A(\mu) := \emptyset$\;
	}\label{empty-end}
	$\Aux := $ the set of all vertices of $\D(\A)$\;
	\ForEach{$\eta \in \Aux$ \Where $\eta$ is a root of $\D(\A)$\label{root-start}}{
		$\Vect_\A(\eta) := \{\V_{\eta}(\vec{1})\}$\;
		$\Aux := \Aux \smallsetminus \{\eta\}$ \;
	}\label{root-end}    
	\While{$\Aux \neq \emptyset$}{
		$\eta := $ an element of $\Aux$ such that $\Pre(\eta) \cap \Aux = \emptyset$ \;
		$\Aux := \Aux \smallsetminus \{\eta\}$ \;
		\ForEach{$\mu \in \Pre(\eta)$}{
			\ForEach{$\bv \in \Vect_\A(\mu)$}{
				$\Vect_\A(\eta) := \Vect_\A(\eta) \cup \{\V_{\eta}(\bv)\}$
			}
		}
		
	}
	\caption{Computing the function $\Vect_\A$}
	\label{alg-vectors}
\end{algorithm}

\section{Proofs for Section~\ref{sec-VASS-games}}
\label{app-games-lower}

\subsection{$\PSPACE$ lower bounds of Theorem~\ref{thm-main-games}}

In this section, we prove the $\PSPACE$ lower bounds of Theorem~\ref{thm-main-games}. We use a modified  construction of Lemma~\ref{lem-hardness1} to obtain a reduction from the QBF problem. Intuitively, the only change is that Angel determines the assignment for universally quantified propositional variables. Let us note that in the original construction of \cite{KLV:VASS-Grzegorczyk}, Angel also selected the clause to be checked. Here, we use the construction of Lemma~\ref{lem-hardness1} based on the $\min$ gadgets.
Let 
\[
\psi \quad \equiv \quad \forall x_1 \, \exists y_1 \, \forall x_2 \ldots \exists y_v
\ C_1\wedge \ldots \wedge C_m
\]
be a quantified Boolean formula where each clause $C_i$ contains precisely three literals (recall that the (in)validity of a given quantified Boolean formula is a $\PSPACE$ complete problem).

Consider the VASS game $\A_\psi$ defined by the VASS~Program~\ref{alg-Apsi}, where $k \geq 2$ is a constant. Note that the only difference between $\A_\psi$ and the demonic VASS $\A_\varphi$ (cf.\ the VASS~Program~\ref{alg-Aphi}) is that the valuation for the universally quantified variables is chosen by Angel. Formally, the gadget for D-\textbf{choose} is the same as the one for \textbf{choose} (see (Fig.~\ref{fig-gadgets}), and the gadget for A-\textbf{choose} is also the same except that the newly added \textit{in} state of the gadget is angelic. Note that \emph{all other states} of $\A_\psi$, including the states in gadgets pumping the $x_i,y_i$ counters, are demonic. Using the observations of Lemma~\ref{lem-hardness1}, it is easy to see that 
\begin{itemize}
	\item if $\psi$ is valid, then $\calL_{\A_\psi} \in \Theta(n^{k+1})$ and \(\calC_{\A_\psi}[s_m] \in \Theta(n^{k}) \);
	\item if $\psi$ is not valid, then $\calL_{\A_\psi} \in \Theta(n^{k})$ and \(\calC_{\A_\psi}[s_m] \in \Theta(n) \).
\end{itemize}
From this we immediately obtain the $\PSPACE$ lower bounds of Theorem~\ref{thm-main-games}.

\SetAlgorithmName{VASS Program}{vassprog}{VASS Programs}
\SetAlFnt{\small}
\SetInd{0ex}{1.55ex}
\begin{algorithm}[tb]
	\SetAlgoLined
	\DontPrintSemicolon
	\SetKwInOut{Parameter}{parameter}\SetKwInOut{Input}{input}\SetKwInOut{Output}{output}
	\SetKwData{C}{C}\SetKwData{D}{D}\SetKwData{MX}{M}\SetKwData{f}{f}
	\SetKw{Choose}{choose:}	
	\SetKw{Or}{or}
	\SetKw{And}{and}
	\BlankLine
	$d_2 \leftarrow d_1 * e_1;\ \ d_3 \leftarrow d_2 * e_2; \ \cdots \ ; \ \ d_k \leftarrow d_{k-1} * e_{k-1};$ \;
	\ForEach{$i = 1,\ldots,v$}{
		A-$\Choose$\ \ $x_i  \leftarrow d_k$ \Or  $\bar{x}_i  \leftarrow d_k;$\;
		D-$\Choose$\ \ $y_i  \leftarrow d_k$ \Or  $\bar{y}_i  \leftarrow d_k;$\;
	}
	$s_0  \leftarrow d_k;$\;
	\ForEach{$i = 1,\ldots,m$}{
		D-$\Choose$\ \ $s_i \leftarrow \min(\ell^i_1,s_{i-1})$
		\Or $s_i \leftarrow \min(\ell^i_2,s_{i-1})$
		\Or $s_i \leftarrow \min(\ell^i_3,s_{i-1});$\;
	}
	$f \leftarrow s_m * n$
	\caption{$\A_\psi$}
	\label{alg-Apsi}
\end{algorithm}

\subsection{A Proof of Proposition~\ref{prop-upper-games}}
\label{app-games-upper}
First, let us restate the proposition.

\PropPolyGames*
\bigskip

\noindent
\textbf{Convention.} For notation simplification, we assume that the counter update vector $\bu$ in every transition $(p,\bu,q)$ where $p$ is angelic satisfies $\bu = \vec{0}$ (a transition $(p,\bu,q)$ where $\bu \neq \vec{0}$ can be split into $(p,\vec{0},q')$, $(q',\bu,q)$ where $q'$ is a fresh demonic state).
\smallskip

Let us fix a vertex $([p],L)$ of $\calG_\A$, a vector $\bv \in \N^d_\infty$, and a counter $c_i$ such that $\bv_i \neq \infty$. First, consider the case when $([p],L)$ is a leaf of $\calG_\A$. Then $([p],L)$ is demonic, and can be seen as a strongly connected demonic VASS. The claim follows trivially by applying Lemma~\ref{lem-simult} to $([p],L)$ and $\bv$.

Now suppose that $([p],L)$ is a demonic vertex of $\calG_\A$ with successors $([q_1],L),\ldots,([q_m],L)$.
Note that $([p],L)$ can be seen as a strongly connected demonic VASS after deleting all transitions from/to the states outside $([p],L)$. Let $F$ be a $\bv$-consistent function. By applying Lemma~\ref{lem-simult} to $([p],L)$, the vector $\bv$, and $F$, we obtain an infinite family of Demon's strategies $\Pi$ and a function $S \in 2^{\calO(n)}$ such that the vector $G(n)$ of counter values in the last configuration of $\beta_n[\bv,p,F,\Pi,S]$ satisfies the following for every $\ell \in \{1,\ldots,d\}$:
\begin{itemize}
	\item $G_\ell(n) \in \Theta(n^k)$ if $\bv_\ell \neq \infty$ and $\C^p_{([p],L)}[c_\ell,\bv] \in \Theta(n^k)$;
	\item $G_\ell(n) \in 2^{\Omega(n)}$ if $\bv_\ell = \infty$ or $\C^p_{([p],L)}[c_\ell,\bv] \in 2^{\Omega(n)}$.
\end{itemize}
%
Let $\bu = \V_{([p],L)}(\bv)$, where $\V$ is the function defined in Section~\ref{sec-upper}. Note that $G$ is a \mbox{$\bu$-consistent} function. If $\bu_i = \infty$, we define $\pi_\bv$ as the strategy which behaves like the strategy $\pi_n$ of $\Pi$ for every computation of $\A_L$ initiated in $p\,F(n)$. Clearly, $\mmax[c_i](\Comp_{\A_L}^{\sigma,\pi_\bv}(p\,F(n))) \in 2^{\Omega(n)}$ for every Angel's strategy $\sigma$. If $\bu_i \neq \infty$, for every $j \in \{1,\ldots,m\}$ we apply the induction hypothesis to $([q_j],L)$, the vector $\bu$, the counter $c_i$, and the $\bu$-consistent function $G$. Thus, we obtain that for every $j \in \{1,\ldots,m\}$ one of the following possibilities holds:

\begin{itemize}
	\item There exists a $k_j \in \N$, a simple locking Angel's strategy $\sigma^j_\bu$ in $\A_{L}$ and a Demon's strategy $\pi^j_\bu$ in $\A_{L}$ such that $\sigma^j_\bu$ is independent of $G$ and
	\begin{itemize}
		\item for every Demon's strategy $\pi$ in $\A_{L}$, we have $\mmax[c_i](\Comp_{\A_{L}}^{\sigma^j_\bu,\pi}(q_j\,G(n))) \in \calO(n^{k_j})$;
		\item for every Angel's strategy $\sigma$ in $\A_{L}$, we have
		$\mmax[c_i](\Comp_{\A_{L}}^{\sigma,\pi^j_\bu}(q_j\,G(n))) \in \Omega(n^{k_j})$.  
	\end{itemize}
	\item There is a Demon's strategy $\pi^j_\bu$ in $\A_{L}$ such that for every Angel's strategy $\sigma$ in $\A_{L}$ we have that 
	$\mmax[c_i](\Comp_{\A_{L}}^{\sigma,\pi^j_\bu}(q_j\,G(n))) \in 2^{\Omega(n)}$.
\end{itemize}

\noindent
We distinguish two cases:

\begin{itemize}
	\item There is $j \in \{1,\ldots,m\}$ such that the second possibility holds. Let $\pi_\bv$ be a Demon's strategy such that, for an initial configuration $p\,F(n)$, $\pi_\bv$ starts by simulating the strategy $\pi_n$ of $\Pi$ until a configuration with the $G(n)$ vector of counter values is reached. Then, $\pi_\bv$ takes the shortest path to a configuration $q_j \bw$, and then switches to simulating $\pi^j_\bu$ for the initial configuration $q_j G(n')$ where $n' \in \N$ is the largest number such that $G(n') \leq \bw$.    
	The properties of $\pi^j_\bu$ (see above) imply that 
	$\mmax[c_i](\Comp_{\A_{L}}^{\sigma,\pi_\bv}(p\,F(n))) \in 2^{\Omega(n)}$ for an arbitrary Angel's strategy $\sigma$.
	\item For all $j \in \{1,\ldots,m\}$, the first possibility holds. Let $k$ be the maximum of all $k_j$ for $j \in \{1,\ldots,m\}$. Furthermore, let $\sigma_\bv$ be a simple locking strategy which for a computation initiated in $p\, F(n)$ behaves like the simple locking strategy $\sigma_{\bu}^j$ when the computation enters a control state of $([q_j],L)$. Note that a computation initiated in $p\, F(n)$ cannot visit an angelic control state before leaving the component $([p],L)$, and the decisions taken by simple locking strategies may depend on the sequence of previously visited components of $\calG_\A$. Hence, $\sigma_\bv$ is a simple locking strategy. For every Demon's strategy $\pi$, we have that $\mmax[c_i](\Comp_{\A_{L}}^{\sigma_\bv,\pi}(p\,F(n))) \in \calO(n^{k})$ by our choice of $k$ and the properties of $\sigma_{\bu}^j$ strategies (see above). Now consider the Demon's strategy $\pi_\bv$ defined in the same way as in the previous case, where $j$ is the index such that $k = k_j$. Then, the properties of $\pi^j_\bu$ imply that 
	$\mmax[c_i](\Comp_{\A_{L}}^{\sigma,\pi_\bv}(p\,F(n))) \in \Omega(n^k)$ for an arbitrary Angel's strategy $\sigma$.	 
\end{itemize}

Finally, suppose that $([p],L)$ is angelic. Then $[p] = \{p\}$, and let $(p,\vec{0},q_1),\ldots,(p,\vec{0},q_m)$ be the outgoing transitions of $p$ (see the Convention above). For every $j \in \{1,\ldots,m\}$, let $L_j = L \cup \{(p,\vec{0},q_j)\}$. 
Let $F$ be a $\bv$-consistent function. By induction hypothesis, for every $j \in \{1,\ldots,m\}$, one of the following possibilities holds:
\begin{itemize}
	\item there exists a $k_j \in \N$, a simple locking Angel's strategy $\sigma^j_\bv$ in $\A_{L_j}$ and a Demon's strategy $\pi^j_\bv$ in $\A_{L_j}$ such that $\sigma^j_\bv$ is independent of $F$ and
	\begin{itemize}
		\item for every Demon's strategy $\pi^j$ in $\A_{L_j}$, we have $\mmax[c_i](\Comp_{\A_{L_j}}^{\sigma^j_\bv,\pi^j}(p\,F(n))) \in \calO(n^{k_j})$;
		\item for every Angel's strategy $\sigma^j$ in $\A_{L_j}$, we have
		$\mmax[c_i](\Comp_{\A_{L_j}}^{\sigma^j,\pi^j_\bv}(p\,F(n))) \in \Omega(n^{k_j})$.  
	\end{itemize}
	\item there is a Demon's strategy $\pi^j_\bv$ in $\A_{L_j}$ such that for every Angel's strategy $\sigma^j$ in $\A_{L_j}$ we have that 
	$\mmax[c_i](\Comp_{\A_{L_j}}^{\sigma^j,\pi^j_\bv}(p\,F(n))) \in 2^{\Omega(n)}$.
\end{itemize}
Strictly speaking, the induction hypothesis applies to the initial configurations $q_j\, F(n)$, but  since $p$ is a control state of $A_{L_j}$ and $(p,\vec{0},q_j)$ is the only out-going transition of $p$ in $\A_{L_j}$, the above claim follows immediately.

Let us first assume that there exists $j$ such that the first possibility holds. Then, we fix a $j$ such that $k_j$ is minimal, and we put $k = k_j$. Consider a simple locking strategy $\sigma_\bv$ in $\A_L$ which starts by locking the transition $(p,\vec{0},q_j)$ and then proceeds by simulating the simple locking strategy $\sigma^j_{\bv}$. For every Demon's strategy $\pi$, we have that $\Comp_{\A_L}^{\sigma_\bv,\pi}(p\,F(n))$ is the \emph{same} computation as $\Comp_{\A_{L_j}}^{\sigma^j_{\bv},\pi}(p\,F(n))$. Hence, $\mmax[c_i](\Comp_{\A_L}^{\sigma_\bv,\pi}(p\,F(n))) \in \calO(n^k)$ by induction hypothesis.

Now consider a Demon's strategy $\pi_\bv$ in $\A_L$ defined as follows. Let $q_1 \bv_1,\ldots,q_\ell \bv_\ell$ be a computation in $\A_L$ initiated in a configuration $p\, F(n)$ such that $q_\ell$ is demonic. Furthermore, let $\alpha \equiv q_1 \tran{\bu_1} q_2 \tran{\bu_2} \cdots \tran{\bu_{\ell-1}} q_\ell$ be the associated path in $\A_L$. The \emph{mode} of $\alpha$ is the $j$ such that $(p,\vec{0},q_j)$ is the \emph{last} outgoing transition of $p$ occurring in $\alpha$. The \emph{$j$-th} projection of $\pi$, denoted by $\pi_j$, is then obtained by concatenating all subsequences of $\alpha$ initiated by the transition $(p,\vec{0},q_j)$ and terminated either by the next occurrence of $p$ or by $q_\ell$. Consider the computation $\gamma_j$ in $\A_{L_j}$ obtained by performing $\pi_j$ from the initial configuration $p\, \lfloor F(n)/m \rfloor$, where $\lfloor F(n)/m \rfloor(i) = \lfloor F_i(n)/m \rfloor$. The transition selected by $\pi_\bv$ after performing the computation $q_1 \bv_1,\ldots,q_\ell \bv_\ell$ is the transition selected by $\pi^j_\bv$, after performing~$\gamma_j$. Intuitively, $\pi_\bv$ ``switches'' among the strategies $\pi^1_\bv,\ldots,\pi^m_\bv$ according to the current mode, and thus simulates computations of $\A_{L_1},\ldots,\A_{L_m}$ initiated in $p\, \lfloor F(n)/m \rfloor$. The computation first reaching a terminal configuration is simulated \emph{completely}. Hence, there exists $j$ such that 
$\Comp_{\A_{L}}^{\sigma,\pi_\bv}(p\,F(n))$ ``subsumes''  $\Comp_{\A_{L}}^{\sigma^j,\pi^j_\bv}(p\,\lfloor F(n)/m \rfloor)$, where $\sigma^j$ is the ``projection'' of $\sigma$ into $\A_{L_j}$. This implies $\mmax[c_i](\Comp_{\A_{L}}^{\sigma,\pi_\bv}(p\,F(n))) \in \Omega(n^{k})$ by induction hypothesis and our choice of~$k$.

Finally, assume that the second possibility holds for \emph{all} $j \in \{1,\ldots,m\}$. Then, we construct a Demon's strategy $\pi_\bv$ in the same way as above, and conclude (by the same reasoning) that 
$\mmax[c_i](\Comp_{\A_{L}}^{\sigma,\pi_\bv}(p\,F(n))) \in 2^{\Omega(n)}$ for an arbitrary Angel's strategy $\sigma$.

\end{document}